\documentclass[acmsmall, screen, nonacm]{acmart}
\usepackage{natbib}

\usepackage{xcolor}

\usepackage{proof}
\usepackage{verbatim}
\usepackage{dapq}
\usepackage{tikz}
\usepackage{chngcntr}
\usepackage{xspace}
\usepackage{quantikz}
\usepackage{braket}
\usepackage[nameinlink]{cleveref}

\AddToHook{cmd/appendix/before}{%
    \crefalias{section}{appendix}%
    \crefalias{subsection}{appendix}
}

\AddToHook{env/example/begin}{\crefalias{theorem}{example}}
\AddToHook{env/lemma/begin}{\crefalias{theorem}{lemma}}
\AddToHook{env/proposition/begin}{\crefalias{theorem}{proposition}}

\newcommand{\nicettfamilysize}{\footnotesize}

\newcommand{\nicettfamily}{%
  \nicettfamilysize\ttfamily%
}%

\lstdefinelanguage{Beluga}{%
  morecomment = [l][\color{gray}]{\%},
  morekeywords = [1]{
    LF, type, ctype, schema, rec, proof, as, case, by, of, unbox, inductive, %
    stratified, some, block, total, mlam, fn, let, in, %
    intros, solve, msplit, suffices, split, by, as, invert, case, unboxed, %
    impossible, undo, toshow, %
    strengthen, %
  },%
  keywordstyle = [1]{\bfseries},%
  morekeywords = [2]{
    unit, arr, c, lam, app, v_lam, v_c, s_app_1, s_app, s_beta, next, refl, halts/m, %
    Unit, Arr, t_lam, t_app, t_c,
    fwd, close, wait, out, inp, pcomp, inl, inr, choice,
    l_fwd1, l_fwd2, l_close, l_wait, l_out, l_inp, l_inl, l_inr, l_choice,
    l_wait2, l_out2, l_out3, l_inp2, l_pcomp1, l_pcomp2, l_inl2, l_inr2,
    l_choice2,
    wtp_fwd, wtp_close, wtp_wait, wtp_out, wtp_inp, wtp_pcomp, wtp_inl, wtp_inr,
    wtp_choice,
    cut1, cut2,
    pair,
    m/u, m/q, m/l, m/w, m/c, match, pat/pair,
    msf/var, msf/var/UL, msf/app/1, msf/var/UU, msf/var/UQ, msf/var/UC,
    bmsf/app, bmsf/lam
  },%
  keywordstyle = [2]{\bfseries\color{purple!90}},%
  morekeywords = [3]{
    tp, tm, val, step, steps, halts, Reduce, RedSub, ctx, nctx, oftype, eq, %
    halts_step, bwd_closed,
    name, hyp, dual, proc, pproc, linear, wtp, equiv, mode, pat,
    msf, bmsf
  },%
  keywordstyle = [3]{\bfseries\color{magenta}},%
  alsoletter=/,%
  columns=flexible,%
  sensitive = true,%
  basicstyle=\nicettfamily,%
  mathescape=true,%
  texcl=true ,%
  escapechar={*},
  literate={%
    {→}{$\rightarrow$ }1%
    {⇒}{$\Rightarrow$}1%
    {->}{$\rightarrow$ }1%
    {par}{$\parr$ }1%
    {times}{$\otimes$}1%
    {bot}{$\bot$}1%
    {⊥}{\(\bot\)}1%
    {||}{\(\|\)}1%
    {⇛}{\(\Rightarrow\)}1%
    {∥}{\(\|\)}1%
    {+}{$\oplus$ }1%
    {ℚ}{$\qubit$}1%
    {⊗}{\(\otimes\) }1%
    {⊸}{\(\multimap\) }1%
    {⅋}{\(\parr\) }1%
    {Φ}{\(\Phi\)}1%
    {⊕}{\(\oplus\) }1%
    {↦}{\(\mapsto\)}1%
    {&}{\(\&\) }1%
    {≻}{\(\geq\) }1%
    {\\}{$\lambda$}1%
    {--nu}{$\nu$}1%
    {beta}{{\color{purple!90}$\beta$}}1
    {⊢}{$\vdash\;$}1%
    {|-}{$\vdash\;$}1%
    {sigma}{$\sigma$}1%
    {Gamma}{$\Gamma$}1%
    {Delta}{$\Delta$}1%
    {Δ}{$\Delta$}1%
  }%
}%

\newcommand{\defor}{\; | \;}

\newcommand{\mkred}[1]{{\color{red}#1}}

\newcommand{\PQD}{HyQ\xspace}

\AtEndPreamble{%
  \theoremstyle{acmdefinition}
  \newtheorem{remark}[theorem]{Remark}
  \newtheorem{case}{Case}
  \newtheorem{subcase}{Subcase}
  \counterwithin*{case}{theorem}
  \counterwithin*{subcase}{case}

}

\begin{document}

\title{Staged Hybrid Quantum-Classical Programming}

\author{Chuta Sano}
\orcid{0000-0002-8179-2307}
\affiliation{
  \department{School of Computer Science}
  \institution{McGill University}
  \country{Canada}
}
\email{chuta.sano@mail.mcgill.ca}         

\author{Peng Fu}
\orcid{0000-0002-3123-0867}
\affiliation{
  \department{Computer Science and Engineering Department}
  \institution{University of South Carolina}
  \country{USA}
}
\email{pfu@cse.sc.edu}

\author{Ryan Kavanagh}
\orcid{0000-0001-9497-4276}
\affiliation{
  \department{Département d'informatique}              
  \institution{Université du Québec à Montréal}            
  \streetaddress{201 avenue du Président-Kennedy}
  \city{Montréal}
  \state{QC}
  \postcode{H2X 3Y8}
  \country{Canada}                    
}
\email{kavanagh.ryan@uqam.ca}

\author{Jennifer Paykin}
\orcid{0009-0008-9502-3219}
\affiliation{
  \department{Department of Computer Science}
  \institution{University of Vermont}
  \country{USA}
}
\email{jennifer.paykin@uvm.edu}

\author{Brigitte Pientka}
\orcid{0000-0002-2549-4276}
\affiliation{
  \department{School of Computer Science} 
  \institution{McGill University}         
  \country{Canada}                        
}
\email{bpientka@cs.mcgill.ca}               

\begin{abstract}
We capture \emph{hybrid quantum-classical computing}. These systems
consist of a classical control system that sends quantum circuits and
receives measurement results from a quantum co-processor. Such systems
allow us to model algorithms that require the classical control system
to generate quantum circuits on the fly, potentially based on prior
measurement results. This is challenging, as the classical control
system must generate quantum circuits that manipulate live quantum
states, e.g., for quantum error correction. 
In this setting, the classical control system is generating further quantum circuits while the quantum co-processor is internally maintaining the state of the live qubits; this is not ideal, since this is not only costly, in the sense of energy consumption, but also introduces additional sources of noise to the live qubits.
Thus, we want to minimize the latency between receiving measurement results
and sending the next quantum circuit to be executed by pre-computing
quantum circuits.

We introduce \PQD (pronounced \textit{haiku}), a multi-modal language
based on adjoint logic that pre-generates quantum circuits before
executing a hybrid quantum-classical program. We achieve this 
by separating our semantics into two distinct stages: 
1) compile-time generation of quantum circuits and classical runtime code and 2)
execution of the classical runtime code that instruments the quantum
co-processor. 
%
This separation between stages allows us to formally guarantee that
all circuit-generation logic occurs before the instrumentation logic,
i.e., the actual runtime, and minimizes the idling of the quantum
co-processor at runtime. 
We give a type system, a circuit-normalization semantics for the
compile-time stage, which eagerly performs all circuit-generation
logic, and a runtime semantics for \PQD that corresponds to the second
stage. We prove type preservation and progress for both semantics.
We further mechanize a superset of \PQD in the Beluga proof assistant.
\end{abstract}

\date{\today}
\maketitle{}

\section{Introduction}

A promising style of quantum languages, such as Quipper~\cite{Green13pldi} and QWIRE~\cite{Paykin17popl}, devise a functional language that generates and runs quantum circuits.
This integration of a functional language, intended to run on a classical processor, with quantum circuits, intended to run on a quantum co-processor, falls in line with the well-studied QRAM model~\cite{Knill1996tr}.
In this model, a classical control system is responsible for performing ``ordinary'' computations and delegates quantum computations to a quantum co-processor by sending it quantum-specific instructions in the form of quantum circuits.
The quantum co-processor then returns probabilistic measurement results to the classical system, which in turn may determine new circuits to pass to the quantum co-processor for execution, and so on.
In its simplist form, each call to the quantum device starts from a blank slate, and newly generated circuits never depend on any existing quantum states from prior iterations.

\begin{figure}[h!]
\begin{tikzpicture}
  \node[draw, rectangle, minimum width=3.5cm, minimum height=1.5cm] (C) {Classical Control System};
  \node[draw, rectangle, minimum width=3.5cm, minimum height=1.5cm, right=5cm of C] (Q) {Quantum Co-processor};

  \draw[->, thick] ($(C.east)+(0,0.6)$) .. controls +(.8,0) and +(-.8,0) .. 
    node[midway, above] {Circuit Commands}
    ($(Q.west)+(0,0.6)$);
  \draw[->, thick]  (Q.west) --
    node[midway, above] {Measurement Results}
    (C.east);
  \draw[->, thick] ($(C.east)+(0,-0.6)$) .. controls +(.8,0) and +(-.8,0) ..
    node[midway, above] {\ldots}
    ($(Q.west)+(0,-0.6)$);
\end{tikzpicture}
  \caption{The basic QRAM model of quantum computation. The classical control system generates quantum circuits and passes them to the quantum co-processor for execution. The quantum co-processor returns measurement results to the classical control system, which may in turn generate new circuits to pass to the quantum co-processor.}
\label{fig:intro-qram}
\end{figure}
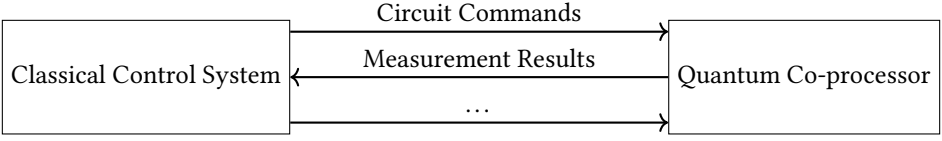

\Cref{fig:intro-qram} depicts a simplified view of this model, where the classical control system cannot depend on any live quantum states from the quantum co-processor when generating subsequent circuits.
In contrast, a challenge arises when an algorithm requires real-time hybrid quantum programs~\cite{lubinski2022advancing}, also referred to as \emph{adaptive measurement}~\cite{raussendorf2001oneway} or \emph{dynamic lifting}~\cite{Green13pldi,Fu23popl}. These terms refer to when the choices of quantum operations depend on measurement results, all within the limited coherence lifetime of the quantum computer. 
An important example is quantum error correction and lattice surgery, where repeated measurements are used to detect and correct errors as they occur. Other examples include distributed quantum algorithms such as quantum teleportation; repeat-until-success circuits; and measurement-based quantum computing---all of which involve dynamic corrections to the state of a live quantum system.

The term dynamic lifting in particular is used in the context of quantum circuit generation and compilation, where a compiler must balance the expressivity of dynamic control with the cost of compiling and optimizing the resulting hybrid program.
Past formulations of dynamic lifting~\cite{Fu23popl}
generate quantum circuits on-the-fly based on prior measurement results.
However, this runtime circuit generation introduces a non-trivial runtime cost, and importantly, is a cost that is incurred on the quantum co-processor. 

As an example, consider the following pseudocode that generates a quantum circuit based on the measurement result of some qubit $q1$:

\begin{lstlisting}
let q1, q2 = some_initialization_logic
if (measure q1) then % generate next circuit for the quantum co-processor
else % generate a different circuit for the quantum co-processor
\end{lstlisting}
In either branch of the conditional, the classical control system must generate a quantum circuit to pass to the quantum co-processor, and during this classical computation, the quantum co-processor must maintain the state of $q2$ (and any other live qubits).

However, compiling quantum programs to executable quantum circuits is expensive. Compilation pipelines include aggressive optimization passes, transpilation to a target gateset, and placement and routing passes to ensure the circuit operations respect architecture connectivity constraints. As a result, compilation is orders of magnitude slower than the lifetime of qubits on state-of-the-art quantum computers.
These algorithms are therefore challenging to implement and reliably execute, because qubits are noisy and can decohere over the time it takes for the classical processor to generate the next quantum circuit to execute.

Thus, we have a trade-off: on one hand, we can have smaller binaries that contain runtime circuit generation logic but incur a runtime cost that is indirectly paid by the quantum co-processor, and on the other hand, we can have larger binaries that pre-compute circuits at compile time but minimizes time that the quantum co-processor is ``idling''.
This is a trade-off that, we argue, is often worth making due to the fact that quantum co-processors are far more expensive and limited than classical resources such as memory.
Thus, we are interested in a language that ensures that all quantum circuits are \textbf{pre-generated} at compile-time to minimize the idling of the quantum co-processor at runtime.

This idea of pre-generating quantum circuits heavily draws on \emph{quantum kernels} in programming languages such as QCor~\cite{mintz2020qcor} and the Intel Quantum SDK~\cite{khalate2022llvm}. These systems compile high-level hybrid quantum-classical programs into separate quantum kernels representing the circuits to run on the quantum machine, which are selected between by classical logic to be run on the classical machine.
As far as we are aware, no prior work has given a formal presentation of such programming paradigms, such as a formal semantics, proof of correctness, or type system.
The Intel Quantum SDK, for instance, relies on users to syntactically separate quantum kernels but lacks complete compile-time checking; this results in subtle issues where the compiler incorrectly orders classical logic before its intended quantum logic~\cite{intel_quantum_sdk_issues}.

In this paper, we introduce \PQD, a multi-modal language based on adjoint logic~\cite{Pruiksma18, Jang24fscd} where we separate
our execution into two stages: (1) circuit-generation time, where quantum circuits are generated and normalized; and (2) runtime, where a minimalistic classical control program instruments the quantum computation.
Broadly speaking, \PQD consists of a functional metaprogramming language, a quantum circuit description language, and a classical runtime language. The metaprogramming language allows programmers to programmatically generate quantum circuits and classical runtime code, which embeds these generated quantum circuits.
The quantum circuit language is a first-order linear $\lambda$-calculus with uninterpreted gate constants where functionals correspond to quantum circuits and gates.
The classical runtime language is a minimalistic functional language that (1) embeds quantum circuits; (2) request measurements on qubits; and (3) can branch based on measurement results.

From the design perspective, the major departing feature of \PQD from prior works is the language-level separation of ``classical computation for circuit generation (stage 1),'' which is run on a classical computer without restrictions, and ``classical computation for runtime (stage 2),'' which are run on a low-power classical control system.
This two-staged separation enables us to cleanly distinguish between classical logic that generates quantum circuits and classical logic that instruments the quantum co-processor's classical control, which notably enables us to enforce that all quantum circuit generation logic \textbf{occurs entirely before} the instrumentation logic, i.e., the runtime involving the quantum co-processor.
This is formally captured through \PQD having two distinct operational semantics: a circuit-normalization semantics that performs all circuit generation logic, ensuring that all quantum circuits are pre-generated, and a runtime semantics that executes the classical control program that instruments the quantum co-processor.

Another advantage of our modular design is that we can easily extend each language with additional features without affecting the other languages. In particular, extensions to the (classical) metaprogramming language will not affect the available features of the (classical) runtime language, which is intended to be minimalistic to model realistic systems.

In short, we make the following contributions:
\begin{itemize}
  \item We introduce \PQD, a foundational calculus that captures the emerging paradigm of pre-generating quantum circuits before executing a hybrid quantum-classical program.
  \item We formalize \PQD by providing its statics and two semantics corresponding to the two stages we explain previously.
    The first stage is formalized by a circuit-normalization semantics, and the second stage is formalized by a runtime semantics. We prove type preservation and progress for both semantics, and we also prove a normal-forms theorem for the circuit-normalization semantics, proving that our language indeed pre-generates all quantum circuits during circuit-generation time.
  \item We mechanize a superset of \PQD, alongside its typing and two semantics, in the Beluga proof assistant. We use an intrinsically-typed encoding, thereby obtaining (partial) type preservation for free, and we fully mechanize the progress proofs of both semantics.
\end{itemize}

\section{Gentle Introduction to \PQD through Examples}
We first informally introduce features of \PQD through motivating examples.
Fundamentally, \PQD is a collection of three languages (or layers) that serve different purposes: a quantum layer that encodes quantum circuits, a classical runtime layer that encodes the classical component that interacts with the quantum layer, and a functional metaprogramming layer that generates and manipulates both quantum and classical runtime programs.
Our design is given in \Cref{fig:mot-pqd}, which contrasts with the QRAM model given in \Cref{fig:intro-qram} where the classical control system is responsible for generating quantum circuits and passing them on the fly. In \PQD, all quantum circuits that the classical control system can send are pre-generated before execution by the metaprogramming language.

\begin{figure}[h!]
\begin{tikzpicture}[>=Latex, every node/.style={font=\sffamily},
  box/.style={draw, rectangle, minimum width=2cm, minimum height=1.5cm, align=center}]
  \node[box, thick, anchor=south] (C) at (0,0) {Classical\\Runtime};
  \node[box, thick, anchor=south] (Q) at (6.5, 0) {Quantum\\Co-processor};
  \node[draw, thick, anchor=south, rectangle, minimum width=3cm, minimum height=3cm, align=center] (M) at (-4.5, 0) {Metaprogramming\\Language};
  \draw[->] ($(C.east)+(0,0.6)$) .. controls +(.8,0) and +(-.8,0) ..
    node[midway, below] {Circuit Commands} ($(Q.west)+(0,0.6)$);
  \draw[->] ($(Q.west)+(0,-0.6)$) .. controls +(-.8,0) and +(.8,0) ..
    node[midway, above] {Measurement Results}
    ($(C.east)+(0,-0.6)$);

  \draw[->, thick, shorten >=1pt] ($(M.north east)-(0,0.20)$) -| ($(C.north)+(0,0.15)$)
    node[midway, above, xshift=2cm] {Generate Hybrid Program};
  \draw[->, thick, shorten >=1pt] ($(M.north east)-(0,0.20)$) -| ($(Q.north)+(0,0.15)$);
\end{tikzpicture}
\caption{A high-level view of \PQD, where the metaprogramming language generates both classical runtime code and quantum circuits, which are then executed by the classical control system and quantum co-processor as in \Cref{fig:intro-qram}. Importantly, all circuit payloads are pre-generated by the metaprogramming language, so the classical control system does not compute any quantum circuits.}
\label{fig:mot-pqd}
\end{figure}
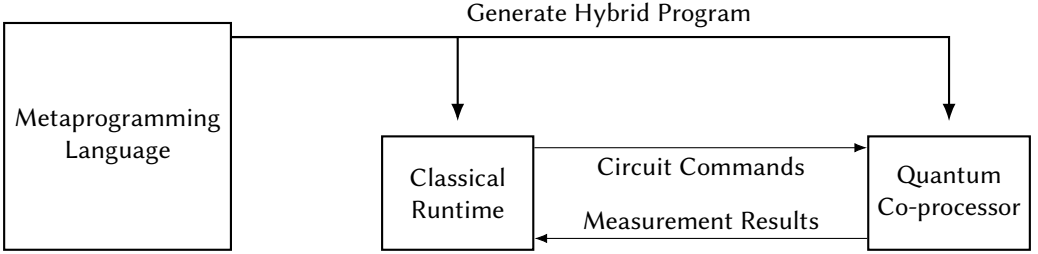

\subsection{Circuit Generation}
The core programming discipline we capture is that of a metaprogramming language that generates hybrid quantum-classical programs.
One component of these programs are the quantum circuits, which are encoded as terms in a first-order linear $\lambda$-calculus.
In this encoding, quantum gates correspond to uninterpreted function constants that map (tuples of) qubits to (tuples of qubits).
Linearity ensures that quantum data cannot be duplicated nor discarded; preventing duplication is necessary to capture the no-cloning theorem, and preventing discarding ensures that the quantum co-processor knows explicitly when it can re-use its physical qubits.

\begin{example}
\label{ex:bell}
A classic example of a quantum circuit is the Bell circuit, which entangles two qubits, both assumed to be in the basis states, into a so-called Bell pair.
A diagramatic representation of the circuit is given in \Cref{fig:mot-bell}.
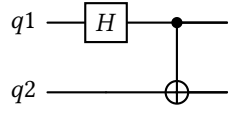
\begin{figure}[t]
\begin{quantikz}
  \lstick{$q1$} & \gate{H} & \ctrl{1} & \qw \\
  \lstick{$q2$} &  & \targ{} & \qw
\end{quantikz}
\caption{Circuit diagram of the Bell circuit, operating on input qubits $q1$ and $q2$. It first applies $H$ to $q1$ and then a $CNOT$ gate to $q1$ and $q2$, with $q1$ as the control.}
\label{fig:mot-bell}
\end{figure}
The following program captures this circuit in the metaprogramming language as a suspended quantum circuit:
\begin{lstlisting}
bell_circ = susp(fn x. let <q1, q2> = x in CNOT <H q1, q2>)
\end{lstlisting}
where \lstinline|<_, _>| is the constructor (and pattern) for (multiplicative) pairs.
\end{example}

The \lstinline|susp| constructor is a \emph{language barrier} that enables the metaprogramming language to capture and suspend quantum circuit language code.
The inner function of type $\qubit \otimes \qubit \multimap \qubit \otimes \qubit$, where $\qubit$ is a base type denoting \emph{qubits}, describes the circuit itself where function application corresponds to gate application.
Both \lstinline|H| and \lstinline|CNOT| are gate constants in the quantum circuit language, which are of type $\qubit \multimap \qubit$ and $(\qubit \otimes \qubit) \multimap (\qubit \otimes \qubit)$ respectively.
In the remaining examples, our language assumes a universal set of gate constants with the appropriate types.

These circuits can be programmatically composed; for example, we can re-use the prior implementation of \lstinline|bell_circ| to create a circuit that operates on four qubits to run two Bell circuits in parallel.
\begin{lstlisting}
parallel_bell = susp(fn x. let <q1, q2, q3, q4> = x in
  <force(bell_circ) <q1, q2>,   % call bell_circ on q1, q2
   force(bell_circ) <q3, q4> >) % call bell_circ on q3, q4
\end{lstlisting}
Here, \lstinline|force| is the eliminator for \lstinline|susp|, and it enables the quantum circuit code to refer back to the metaprogramming language.
We have taken some syntactic liberties above for presentation purposes; in particular, we matched on the input $x$ as a $4$-tuple despite only having binary products.
In the remaining examples, we will similarly take syntactic liberties for presentation purposes.

The ability for quantum circuits to refer to metaprogramming language terms enables programmatically generating families of quantum circuits that leverage high-level features of the metaprogramming language.

\begin{example}
We extend our metaprogramming language with recursion and lists to define a function
that takes as input a list of $n$-ary circuits and ``folds'' it to generate a single $n$-ary circuit that applies all the circuits in the list in sequence for some fixed $n$:
\label{ex:fold}
\begin{lstlisting}
fold_circuits cs = match cs with
  | nil -> susp(fn x. x) % generate the identity circuit
  | c::cs' -> susp(fn x. (force (fold_circuits cs')) ((force c) x))
\end{lstlisting}
\end{example}
The base case gives the identity circuit that does nothing on its input.
For the recursive case, it applies the head of the list \lstinline|c| to the input \lstinline|x| and then applies the result to the circuit generated by recursively folding the tail of the list \lstinline|cs'|.

The metaprogramming language is not limited to generating quantum circuits with input wires.
It can also generate quantum circuits that generate qubits from scratch through an \lstinline|init| constant, which, like other gates, is a constant of function type: $1 \multimap \qubit$, where $1$ is the (multiplicative) unit type.
For example, the following program generates a quantum circuit that creates a Bell pair without any input qubits:
\begin{lstlisting}
bell_pair = susp((force bell_circ) <init <>, init <>>)
\end{lstlisting}
where \lstinline|<>| is the trivial constructor of the unit type $1$.

\subsection{Classical Code Generation}
As we explain previously, quantum systems consist of a classical control system and a quantum co-processor.
We model this control system through another language, which is a minimalistic functional language that (1) embeds quantum circuits; (2) request measurements on qubits; and (3) can branch based on measurement results.

The following simple program generates a classical runtime program that runs the \lstinline|bell_pair| circuit to obtain two qubits \lstinline|q1, q2| and measures \lstinline|q1|:
\begin{lstlisting}
let down <q1:qubit, q2:qubit> = down (force bell_pair) in
if (measure q1) then % do something
else % do something else
\end{lstlisting}
The classical runtime language uses \lstinline|down| to refer to quantum circuits and their results; in particular, the inner term of the first line, \lstinline|force bell_pair|, is a quantum language term of type $\qubit \otimes \qubit$.
It then decomposes the resulting pair of qubits into \lstinline|q1| and \lstinline|q2|, which, in future examples, we will inline together with the \lstinline|let down| eliminator.
Finally, it \emph{measures} the first qubit \lstinline|q1| through the \lstinline|measure| primitive, which converts a qubit into a classical bit, which can then be used in a conditional.

\begin{example}
A simple program that uses dynamic lifting is the quantum teleportation protocol, which we implement as a function in the classical runtime language. Teleportation takes as input a reference to a live qubit and outputs a reference to a live qubit.
The diagrammatic representation of the protocol (see \Cref{fig:mot-teleport}) indicates that the last two last two operators $\gsty{Z}$ and $\gsty{X}$ are only applied if the corresponding measurement results produce a $1$. 
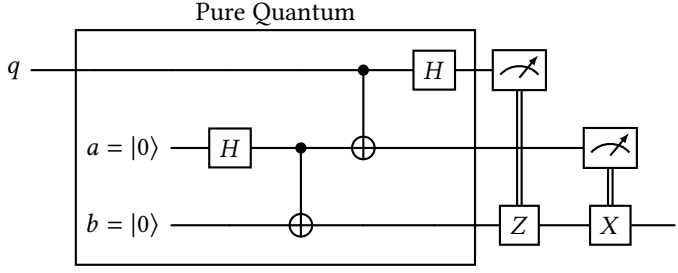
\begin{figure}[t]
\begin{quantikz}[wire types={q,n,n}]
  \lstick{$q$} &[1em] \gategroup[3, steps=7]{Pure Quantum} & & & &  & \ctrl{1} & \gate{H} & \meter{}\wire[d][2]{c}    \\
  & & & \lstick{$a = \ket{0}$} & \gate{H} \setwiretype{q}  & \ctrl{1} & \targ{} &   & & \meter{}\wire[d][1]{c}  \\
  & & & \lstick{$b = \ket{0}$} & \setwiretype{q}  & \targ{} &  & &\gate{Z}  &  \gate{X} &
\end{quantikz}
\caption{A diagrammatic representation of the quantum teleportation protocol, where $q$ is the qubit to be teleported, and $a, b$ are auxiliary qubits that are assumed to be in the basis state.
The initial purely quantum parts of the diagram have been separated for presentation purposes, and the remaining two operations on $b$ are classically controlled by the measurement results of $q$ and $a$, respectively. The protocol results in transporting $q$'s state to $b$.}
\label{fig:mot-teleport}
\end{figure}

The protocol begins with a purely quantum circuit that initializes $a, b$ into a Bell pair and applies additional gates to entangle $q$ with $a$ and $b$ and well.
The subsequent instructions in the diagram rely on \emph{measuring} the qubits $q$ and $a$, and applying $\gsty{Z}$ and $\gsty{X}$ gates to $b$ based on those measurement results.
While the individual gate applications are quantum, the decision of whether to apply them is classical, and thus we implement the latter part of the protocol in the classical runtime language.
\begin{lstlisting}[numbers=left, xleftmargin=5.0ex]
teleport = fn (q_ptr: down qubit).
  let down (q:qubit) = q_ptr in
  let down <a:qubit, b:qubit> = down (force bell_pair) 
  let down <q:qubit, a:qubit> = down (CNOT <q, a>) in
  let down (q:qubit) = down (H q) in
  let down (b:qubit) =
    if (measure q) then down (Z b) else down b
  in
  if (measure a) then down (X b) else down b
\end{lstlisting}
Lines $3-5$ implement the initial purely quantum computation, $6-8$ the first classically-controlled operation on $b$, and $9$ the final classically-controlled operation on $b$.
Since classical runtime programs capture actual computations that are run on a hybrid system, we do not expect this implementation, which is a function, to run on its own, but rather be composed in a larger program that provides the necessary input qubit reference.
\end{example}

A more complicated example is a repeat-until-success program~\cite{Paetznick14qic}, a non-deterministic hybrid program that can approximate arbitrary gates.
We implement the $V_3$ gate based on the implementation by \citet{Fu23popl} in \Cref{fig:ex-v3}.
\begin{figure}[t]
\begin{lstlisting}[numbers=left, xleftmargin=5.0ex]
fix v3 : (down qubit) -o (down qubit).
fn x : (down qubit).
let down <a1, a2> = down(
  let <q1, q2> = CNOT <T_inv (H (init ())), H (init ())> in
  <H (T q1), q2>)
in
if (measure a1)
then let down () = down (discard a2) in v3 x % try again
else let down q = x in
  let down <a2', q'> = down(
    let <q1, q2> = CNOT <a2, Z (T q)> in
    <q1, H (T q2)>)
  in
  if (measure a2')
  then v3 (down (Z q')) % try again
  else down q' % success
\end{lstlisting}
\caption{Implementation of the $V_3$ gate in the classical runtime language. It has type $\dqc \qubit \multimap \dqc \qubit$, i.e., a function that takes a reference to a live qubit and outputs another reference to a live qubit. The \lstinline|fix| construct is the usual fixed-point construct that binds \lstinline|v3| for the body, enabling the modeling of the ``retry'' mechanism in this algorithm.}
\label{fig:ex-v3}
\end{figure} 

The metaprogramming language can embed these classical runtime programs as code through a \lstinline|down| constructor (which, we remark, is distinct from the \lstinline|down| constructor used in the classical runtime language to embed quantum terms).
The use of \lstinline|down| to capture classical runtime programs as opposed to \lstinline|susp| to capture quantum circuits stems from a technical reason that we elaborate on in \Cref{ssec:instrumenting}, but for now, we can think of the two constructs as serving analogous purposes.

One could lift \Cref{fig:ex-v3} to the metaprogramming language by wrapping its entire body in the metaprogramming language's \lstinline|down| constructor:
\begin{lstlisting}
v3_program = down(fix v3. ...)
\end{lstlisting}

The results of these captured classical runtime programs can, like in the quantum circuit case, be composed with other classical runtime programs.
For instance, if we have some metaprogramming term \lstinline|v3_input| that generates some classical runtime program that outputs a reference to a live qubit, we can compose the two programs directly in the metaprogramming language as follows:
\begin{lstlisting}
let down f = v3_program in let down x = v3_input in
down(f x)
\end{lstlisting}

In summary, our overall design of a separate classical runtime language that the metaprogramming language can manipulate enables the following features:
\paragraph{1. Separate Generation of Quantum Circuits}
In \Cref{fig:ex-v3}, we inlined all quantum circuits directly in the classical runtime language, partly due to the relative simplicity of the embedded quantum programs, but also for presentation purposes.
Instead, we can leverage the metaprogramming language to generate the quantum circuits separately and simply refer to them via a \lstinline|force|.
For example, we can rewrite lines 3 to 6 of \Cref{fig:ex-v3}:
\begin{lstlisting}
let down <a1, a2> = down(force some_implementation) in ...
\end{lstlisting}
where \lstinline|some_implementation| is the metaprogramming term that generates the quantum circuit in lines 4 and 5 of \Cref{fig:ex-v3}. This enables, for instance, a library of commonly used families of quantum circuits, to be pre-generated and reused in different classical runtime programs.

\paragraph{2. Families of Classical Runtime Programs}
Just as we can leverage high-level programming features in the metaprogramming language to generate families of quantum circuits, we can also leverage these features to generate families of classical runtime programs.
In particular, the aforementioned repeat-until-success style program are among many examples of hybrid quantum-classical programs that simulate quantum gates.
These could be incorporated as metaprograms that generate particular programs that simulate gates, which can then be programatically used in larger quantum algorithms.
More generally, arbitrary hybrid quantum-classical programs can be captured through the metaprogramming language and manipulated, not only giving a great deal of code reusability but also allowing for domain-specific optimizations.

\PQD is therefore aimed to be a foundation for the following programming pipeline to generate hybrid quantum-classical programs:
\begin{enumerate}
  \item The programmer leverages high-level programming features in the metaprogramming language to metaprogram $Q$ code
    -- this gives libraries for families of (purely) quantum circuits.
  \item The programmer, again, leverages high-level programming features to metaprogram $C$ code, which can include $Q$ code generated in the previous step -- this also gives libraries for hybrid quantum algorithms.
  \item The programmer can execute the generated $C$ code on a hybrid quantum-classical system, where any $Q$ code is executed on the quantum co-processor and the $C$ code is executed on the classical control system.
\end{enumerate}

In the remaining sections, we formalize the design of \PQD that we informally introduced here.

\section{Adjoint Logic as a Foundation for Quantum Programming}
\label{sec:statics}
Adjoint Logic~\cite{Pruiksma18} is a framework for combining multiple logical systems, or \emph{modes}, that have different substructural properties.
In general, an adjoint system consists of a set of modes annotated with their substructural properties (weakening and contraction), their permitted type constructors, and a pre-order on the modes that indicate which modes can depend on assumptions of other modes.
For example, LNL~\cite{Benton94csl} can be seen as a two-moded adjoint system consisting of a linear mode $L$ and an unrestricted mode $U$. The pre-order $U \geq L$ models the restriction that proofs of unrestricted propositions cannot depend on linear assumptions.

It has been given a natural deduction-style presentation~\cite{Jang24fscd} to model a family of programming languages and their interactions, where each mode corresponds to a language, and the shifts between modes (which form an adjoint) allow languages of different modes to embed and refer to each other.

Adjoint logic as a foundation for quantum programming has been explored by Kavanagh et al.~\cite{Kavanagh26ppdp}, where they
gave a recontruction of the quantum programming language Proto-Quipper-M~\cite{Rios18qpl} as a three-moded adjoint system.
Their system supports quantum circuit generation but not execution, and they do not consider a classical runtime language.
Our work builds on this design, though we omit several unnecessary constructs and make some minor notational/syntactic differences; we explain more precisely the differences between our variant and the original Proto-Quipper-A in \Cref{sec:rel-work}.
We proceed by introducing the quantum circuit-generation fragment of \PQD.

\subsection{Quantum Circuit Generation}
The circuit generation fragment of \PQD (see \Cref{fig:pqa}) is a three-moded adjoint system consisting of modes $U$, $L$, and $Q$ with the mode pre-order $U \geq L \geq Q$ and $Q \geq L$.
Intuitively, one can think of $U$ (unrestricted) and $L$ (linear) combined forming the functional metaprogramming language, and $Q$ corresponding to the quantum circuit language.

Type judgments are of the form $\tpm{\Gamma}{m}{M}{A}$, where $m$ ranges over the modes $U, L, Q$, and $\Gamma$ is a context that must only contain assumptions of mode $m$ and greater.
The subscripts $m$ are mostly for presentation purposes and may be omitted when the mode is clear from context.
The adjoint structure allow us to capture the common types and syntactic constructs of these languages; in particular, all three modes have (multiplicative) pairs, unit, and function types, though $Q$ functions are restricted to only be first order since they are intended to capture quantum circuits.
We first give this simple fragment, which captures the base features of each individual mode:
\[
  \infer[app]{\tp{\GammaIII}{m}{\appm{M}{N}{m}}{B_m}}
  {\tp{\Gamma_1}{m}{M_m}{A_m \multimap_m B_m}
  &\tpm{\Gamma_2}{m}{N}{A}}
  \quad
  \infer[lam]{\tp{\Gamma}{m}{\lamm{x}{M}{m}}{A_m \multimap_m B_m}}
  {\tpm{\Gamma, x : A_m}{m}{M}{B}}
\]
\[
  \small
  \infer[var]{\tp{\Gamma_U, x : A_m}{m}{x}{A_m}}{}
  \quad
  \infer[triv]{\tp{\Gamma_U}{m}{\trivm{m}}{1_m}}{}
  \quad
  \infer[letu]{\tp{\GammaIII}{m}{\ltrivm{M}{N}{m}}{A_m}}
  {\tpm{\Gamma_1}{m}{M}{1}
  &\tpm{\Gamma_2}{m}{N}{A}}
\]
\[
  \small
  \infer[pair]{\tp{\GammaIII}{m}{\pairm{M}{N}{m}}{A_m \otimes_m B_m}}
  {\tpm{\Gamma_1}{m}{M}{A}
  &\tpm{\Gamma_2}{m}{N}{B}}
  \quad
  \infer[letp]{\tp{\GammaIII}{m}{\lpairm{x}{y}{M}{N}{m}}{C_m}}
  {\tp{\Gamma_1}{m}{M_m}{A_m \otimes_m B_m}
  &\tpm{\Gamma_2, x : A_m, y : B_m}{m}{N}{C}}
\]
All rules shown above are standard, though we note several things. First, the use of $\Gamma_U$ in the axiom rules ($var$ and $triv$) is notation to denote a context that only contains assumptions of mode $U$, since all other modes do not allow weakening.
Next, the notation $\GammaIII$ denotes a context merge between $\Gamma_1$ and $\Gamma_2$, defined below.
\begin{align*}
  (\Gamma_1, x:A_m) \bowtie \Gamma_2 &= (\Gamma_1 \bowtie \Gamma_2), x:A_m
    \tag{for $m \in \{L, Q\}$}
  \\
  \Gamma_1 \bowtie (\Gamma_2, x:A_m) &= (\Gamma_1 \bowtie \Gamma_2), x:A_m
    \tag{for $m \in \{L, Q\}$}
  \\
  (\Gamma_1, x:A_U) \bowtie (\Gamma_2, x:A_U) &= (\Gamma_1 \bowtie \Gamma_2), x:A_U
\end{align*}

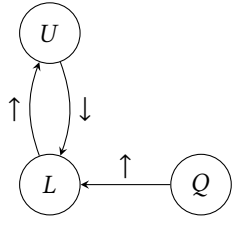
\begin{figure}[t]
\begin{centering}
\begin{tikzpicture}[>=stealth, node distance=1cm]
  \tikzset{circle node/.style={circle, draw, minimum size=8mm, inner sep=0pt, font=\sffamily\bfseries}}
  \node[circle node] (U) at (0,2) {$U$};
  \node[circle node] (L) at (0,0) {$L$};
  \node[circle node] (Q) at (2,0) {$Q$};
  \draw[->, bend left=20] (U) to  node[midway, right]{$\downarrow$} (L);
  \draw[->, bend left=20] (L) to node[midway, left]{$\uparrow$} (U);
  \draw[->] (Q) to node[midway, yshift=7]{$\uparrow$} (L);
\end{tikzpicture}
\end{centering}
\caption{Mode structure of \PQD's circuit generation fragment, where $U \geq L \geq Q$ and $Q \geq L$. The arrows indicate shifts between modes that are syntactically allowed.}
\label{fig:pqa}
\end{figure}

Adjoint logic further provides shifts between modes that govern how terms of one mode can refer to terms of another mode.
While \Cref{fig:pqa} shows all shifts that we syntactically allow, we first introduce the general form of the shift rules and then explain the computational meaning of each of the allowed shifts later.
The $\uparrow$ shift enables the embedding of \emph{lower} modes into \emph{higher} modes.
\[
  \infer[{\uparrow} I]{\Gamma \vdash_n \susp{m}{n}{M_m} : {\uparrow_m^n} A_m}
  {\Gamma \vdash_m M_m : A_m
  &\Gamma \geq n}
  \quad
  \infer[{\uparrow} E]{\Gamma \vdash_m \force{m}{n}{M_n} : A_m}
  {\Gamma \vdash_n M_n : {\uparrow_m^n} A_m
  &\Gamma \geq n}
\]
The notation $\Gamma \geq n$ denotes that all assumptions in $\Gamma$ are of mode $n$ or greater,
which captures the invariant that terms can only depend on assumptions of higher (or equal) modes.
We remark that this mode condition is implicitly satisfied in either the introduction or the elimination rule, depending on whether the rules are read top-down (i.e., as an inductive definition) or bottom-up (i.e., as a type-checking algorithm), but we redundantly include it as premises in both rules to avoid confusion.

Intuitively, we think of $\mathsf{susp}\; M$ as a higher mode \emph{suspending} some code of a lower mode, e.g., code generation, and $\mathsf{force} \; M$ as the lower mode $\emph{forcing}$ the execution of the suspended code generated by the higher mode.

On the other hand, the $\downarrow$ shift enables \emph{lower} modes to refer to \emph{higher} modes.
\[
  \infer[{\downarrow} I]{\Gamma \vdash_m \down{n}{m}{M_n} : {\downarrow_m^n} A_n}
  {\Gamma \vdash_n M_n : A_n
  &\Gamma \geq n}
  \quad
  \infer[{\downarrow} E]{\GammaIII \vdash_n \ldown{n}{m}{x}{M_m}{N_n} : B_n}
  {\Gamma_1 \vdash_m M_m : {\downarrow_m^n} A_n
  &\Gamma_2, x : A_n \vdash_n N_n : B_n
  &\Gamma_1 \geq n}
\]
The mode conditions are included in both rules for claritiy, though as in the $\uparrow$ shift case, they are redundant in one of the rules depending on how the rules are read.

Intuitively, we think of $\mathsf{down}\; M$ as the reference to the \emph{result of computing} $M$, and the primary distinction between the two shifts is that a term inside $\mathsf{susp}$ is treated as syntax and therefore not reduced further until it is forced via a $\mathsf{force}$, while a term inside a $\mathsf{down}$, is always reduced.

We now explain features that are specific to the quantum circuit language ($Q$ mode) and the functional metaprogramming language ($U$ and $L$ modes).

\subsubsection{Quantum Circuit Language}
The quantum circuit language $Q$ is a linear $\lambda$-calculus with the base type $\qubit$, denoting a \emph{qubit}.
It has no constructors nor destructors, although we assume the constants $\gsty{init}: 1_Q \multimap_Q \qubit$ and $\gsty{discard}: \qubit \multimap_Q 1_Q$ to encode the initialization and discarding of qubits, respectively.
Since $Q$ is meant to capture quantum circuits, function spaces in $Q$ are restricted to be first-order, which is formally captured by requiring functions to span over \emph{simple types} $U, S$:
\begin{align*}
  U, S \Coloneqq 1_Q \defor \qubit \defor U \otimes_Q S \qquad
  A_Q, B_Q \Coloneqq U \defor U \multimap_Q S
\end{align*}
The language $Q$ also assumes a global signature of gate constants of the appropriate (function) types, where gates are treated as uninterpreted (function) constants.
For example, the Hadamard gate $\gsty{H}$, which operates on a single qubit, is a symbol of type $\qubit \multimap_Q \qubit$, and the $\gsty{CNOT}$ gate, which operates on two qubits has type $(\qubit \otimes_Q \qubit) \multimap_Q (\qubit \otimes_Q \qubit)$.

\begin{example}
\label{pqa:bell} 
The following is the typing derivation of the Bell circuit given in \Cref{ex:bell}:
\[
  \small
  \infer[lam]{\qtp{\cdot}{\lam{x}{\lpairq{q_1}{q_2}{x}{\app{\gsty{CNOT}}{\pairq{\app{\gsty{H}}{q_1}}{q_2}}}}}{\qubit \otimes \qubit \multimap_Q \qubit \otimes_Q \qubit}}
  {\infer[letp]{\qtp{x: \qubit \otimes_Q \qubit}{\lpairq{q_1}{q_2}{x}{\app{\gsty{CNOT}}{\pairq{\app{\gsty{H}}{q_1}}{q_2}}}}{\qubit \otimes_Q \qubit}}
    {\infer[app]{\qtp{q_1 : \qubit, q_2 : \qubit}{\app{\gsty{CNOT}}{\pairq{\app{\gsty{H}}{q_1}}{q_2}}}{\qubit \otimes_Q \qubit}}
      {\infer[]{\qtp{\cdot}{\gsty{CNOT}}{(\qubit \otimes_Q \qubit) \multimap_Q (\qubit \otimes_Q \qubit)}}
       {}
      &\infer[pair]{\qtp{q_1 : \qubit, q_2 : \qubit}{\pairq{\app{\gsty{H}}{q_1}}{q_2}}{\qubit \otimes_Q \qubit}}
       {
         \infer[app]{\qtp{q_1 : \qubit}{\app{\gsty{H}}{q_1}}{\qubit}}
         {\infer[]{\qtp{\cdot}{\gsty{H}}{\qubit \multimap_Q \qubit}}{}
         &\infer[var]{\qtp{q_1 : \qubit}{q_1}{\qubit}}{}}
         &\infer[var]{\qtp{q_2 : \qubit}{q_2}{\qubit}}{}
       }
      }
}}
\]
We omit the typing derivation of $x : \qubit \otimes_Q \qubit \vdash x : \qubit \otimes_Q \qubit$ as a premise in the $letp$ rule for presentation purposes.
\end{example}

\subsubsection{Functional Metaprogramming Language}
The metaprogramming language consists of the modes $U$ and $L$, respectively corresponding to the unrestricted and linear fragments of the language.
We often use $F$ to range over the two modes in subsequent sections.

As a functional language, $F$ internally contains two shifts $\ulu A_L$ and $\dul A_U$.
Intuitively, terms of type $\ulu A_L$ capture code of type $A_L$ that are closed with respect to $L$-moded assumptions, enabling it to be duplicated and discarded in the unrestricted mode $U$.
On the other hand, terms of type $\dul A_U$ can be thought as ($L$-moded) references to unrestricted values of type $A_U$.
The composition $\dul \ulu A_L$ exactly corresponds to the linear exponential $!A_L$, and
for programming purposes, we can essentially treat $U$ and $L$ as two fragments of a linear $\lambda$-calculus.

The functional language is further equipped with the ability to refer to quantum circuit terms via the $\uql$ shift, which, like $\ulu A_L$, captures quantum circuit terms as code.

\begin{example}
\label{pqa:bell-susp}
Since the example given in \Cref{pqa:bell} is a $Q$-moded term, we can capture it as code in the $L$-mode by using the $\uql$ shift:
\[
  \infer[{\uparrow} I]{\ltp{\cdot}{\suspql{\lam{x}{\lpairq{q_1}{q_2}{x}{\app{\gsty{CNOT}}{\pairq{\app{\gsty{H}}{q_1}}{q_2}}}}}}{{\uql} (\qubit \otimes \qubit \multimap_Q \qubit \otimes_Q \qubit)}}
  {
    \infer[lam]{\qtp{\cdot}{\lam{x}{\lpairq{q_1}{q_2}{x}{\app{\gsty{CNOT}}{\pairq{\app{\gsty{H}}{q_1}}{q_2}}}}}{\qubit \otimes \qubit \multimap_Q \qubit \otimes_Q \qubit}}
  {\cdots}
}
\]
\end{example}

\begin{remark}
If a quantum circuit code of type $\uql A_Q$ is closed with respect to $L$ and $Q$ assumptions, then it can be lifted to the unrestricted mode $U$ as a term of type $\ulu \uql A_Q$. At this point, it can be freely duplicated and discarded in the unrestricted mode and forced back to the $L$ (and then $Q$) mode for re-use.
\end{remark}

\subsection{Adding a Classical Runtime Language}
\label{ssec:statics-classical}
The fragment of \PQD that we gave thus far defines a metaprogramming language $U/L$ that can generate and compose quantum circuit programs $Q$.
Instead of adding further extra-logical primitives to the metaprogramming language to capture runtime features, we continue in the philosophy of adjoint logic and introduce a linear mode $C$ and an unrestricted mode $W$ that, together, model a classical runtime language that interacts with the quantum co-processor (as modeled by mode $Q$).
Indeed, as we mentioned previously, no quantum systems run purely quantum code, so this characterization of a classical runtime language that interacts with the quantum layer models the reality of quantum programming more closely.

\paragraph{Mode $C$}
Classical runtime programs must be able to reference live quantum bits and run pre-generated quantum circuits, thus, we must have $Q \geq C$ (therefore $U \geq C$ and $L \geq C$ by transitivity) and $C$ be linear.
We also require that $C \geq L$ (and therefore $C \geq Q$ by transitivity), which we justify later.
Our design of $C$ is motivated by both real-world constraints and examples to be a minimalistic functional language extended with the ability to delegate the execution of $Q$-moded quantum circuits, request measurements on qubits, branch based on these measurement results, and loop.

\paragraph{Mode $W$}
On the other hand, the unrestricted classical runtime mode $W$ exists solely for the purpose of modeling exponentials, for example, to allow duplicating and discarding the result of measurement, which are just classical bits.
It therefore only contains $\uparrow_C^W A_C$ as its allowed type, and analogously, $\mathsf{susp}_C^W M_C$ as its only term.
Most of the discussion that follows is therefore focused on the design of $C$, since $W$ does not interact with $Q$ in any direct way.

We give a graphical summary of our entire system in \Cref{fig:pqd}.

\begin{figure}[!h]
\begin{centering}
\begin{tikzpicture}[>=stealth, node distance=1cm]
  \tikzset{circle node/.style={circle, draw, minimum size=8mm, inner sep=0pt, font=\sffamily\bfseries}}
  \node[circle node] (U) at (0,2) {$U$};
  \node[circle node] (L) at (0,0) {$L$};
  \node[circle node] (Q) at (0,-2) {$Q$};
  \node[circle node, color=red] (C) at (2,-2) {\mkred{$C$}};
  \node[circle node, color=red] (W) at (2,2) {\mkred{$W$}};
  \draw[->, bend left=20] (U) to  node[midway, right]{$\downarrow$} (L);
  \draw[->, bend left=20] (L) to node[midway, left]{$\uparrow$} (U);
  \draw[->] (Q) to node[midway, right]{$\uparrow$} (L);
  \draw[->, color=red] (Q) to node[midway, yshift=7]{\mkred{$\downarrow$}} (C);
  \draw[->, color=red] (C) to node[midway, yshift=10]{\mkred{$\downarrow$}} (L);
  \draw[->, bend left=20, color=red] (C) to node[midway, right]{\mkred{$\uparrow$}} (W);
  \draw[->, bend left=20, color=red] (W) to node[midway, left]{\mkred{$\downarrow$}} (C);
\end{tikzpicture}
\end{centering}
\caption{All modes and their shifts in \PQD, where $U \geq L \geq Q \mkred{\geq C}$, $\mkred{C \geq Q \geq L}$, and $\mkred{W \geq C}$. The modes and shifts that we add to \Cref{fig:pqa} are colored in \mkred{red}.}
\label{fig:pqd}
\end{figure}
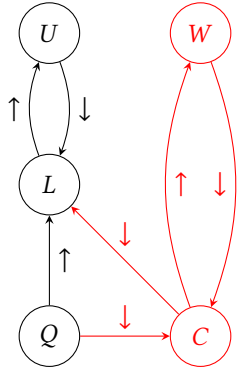

\paragraph{Passing Requests to the Quantum Co-Processor}
We add a \emph{restricted} downshift $\dqc U$ that models a request by the classical control to execute (quantum) instructions on the quantum co-processor.
We only allow $\downarrow$ shifts of \emph{simple types} $U$, since it makes no sense for the quantum co-processor to execute a function.

This is not sufficient however; consider the following incomplete derivation of a program that runs $\mathsf{bell\_pair}$ (see \Cref{pqa:bell-susp}) on two initialized qubits:
\[
  \infer[{\downarrow} E]
  {\ctp{\cdot}
    {\ldownqc{x}
      {\app
        {\forcelq{\mathsf{bell\_pair}}}
        {\pairq
          {\app{\gsty{init}}{\trivq}}
          {\app{\gsty{init}}{\trivq}}}
      }
      {N}
    }
    {B}
  }
  {\infer{\ctp{\cdot}{\app
        {\forcelq{\mathsf{bell\_pair}}}
        {\pairq
          {\app{\mathsf{init}}{\trivq}}
          {\app{\mathsf{init}}{\trivq}}}
      }{{\dqc} (\qubit \otimes_Q \qubit)}}{\cdots}
  & \infer{\ctp{x : \qubit \otimes_Q \qubit}{N}{B}}{\cdots}}
\]
The body $N$ cannot split the two qubits in $x$ since $x$ is a $Q$-moded assumption.
Thus, we must also allow pattern matching on $Q$-moded terms in $C$, which amounts to extending $C$ with $\otimes_Q$ and $1_Q$ elimination rules.
This is made possible by the ordering $Q \geq C$, which allows $C$-moded terms to depend on $Q$-moded assumptions.
\[
  \small
  \infer[letu_C^Q]{\ctp{\GammaIII}{\ltrivq{M_Q}{N_C}}{A_C}}
  {\qtp{\Gamma_1}{M_Q}{1_Q}
  &\ctp{\Gamma_2}{N_C}{A_C}}
  \quad
  \infer[letp_C^Q]{\ctp{\GammaIII}{\lpairq{x}{y}{M_Q}{N_C}}{C_C}}
  {\qtp{\Gamma_1}{M_q}{A_Q \otimes_Q B_Q}
  &\ctp{\Gamma_2, x : A_Q, y : B_Q}{N_C}{C_C}}
\]

We can now proceed in our previous example by using $letp_C^Q$ to split the two qubits in $x$ and run further quantum operations on the individual qubits:
\[
  \infer[letp_C^Q]{\ctp{x: \qubit \otimes_Q \qubit}{\lpairq{q_1}{q_2}{x}{N_C'}}{C_C}}
  {\infer[var]{\qtp{x : \qubit \otimes_Q \qubit}{x}{\qubit \otimes_Q \qubit}}{}
  &\infer{\ctp{q_1 : \qubit, q_2 : \qubit}{N_C'}{C_C}}{\cdots}}
\]

\paragraph{Measurement and Branching}
In quantum systems, a \emph{measurement} is a destructive operation that collapses a qubit into a classical bit that the classical system can use.
We therefore add $\cbit$, denoting a bit, as a base type in $C$ in order to store measurement results, and we add a primitive $\mathsf{measure}$ that takes in a qubit and outputs a bit.
We do not consider any binary operators on bits for simplicity, though they can be easily added.

To allow classical bits to be duplicated and discarded, our measurement primitive gives some $\dwc \ucw \cbit$, which again, can be treated as the encoding of $! \cbit$ in a typical linear type system.
Finally, to capture branching based on measurement results, we add an if-then-else construct to mode $C$ that branches on a bit.
\[
  \infer[meas]{\ctp{\Gamma}{\meas{M_Q}}{\dwc \ucw \cbit}}
  {\qtp{\Gamma}{M_Q}{\qubit}}
  \quad
  \infer[if]{\ctp{\GammaIII}{\ifc{M}{N_1}{N_2}}{A_C}}
    {\ctp{\Gamma_1}{M}{\cbit}
    &\ctp{\Gamma_2}{N_1}{A_C}
    &\ctp{\Gamma_2}{N_2}{A_C}}
\]

\begin{remark}
While re-using and discarding measurement results is important, many algorithms (including our examples thus far) only use measurement results once, making it convenient to have a ``linear'' measurement operation to enable inlining measurements with branching.
This can be derived by the syntactic sugar:
\begin{lstlisting}
measure' M := letd x = measure M in force x
\end{lstlisting}
which justifies the inlining of measurements in our examples thus far.
\end{remark}

\paragraph{Loop}
Hybrid quantum-classical algorithms, such as \Cref{fig:ex-v3},
  require an arbitrary number of iterations and do not necessarily terminate; we model this by adding a fixed-point operator to $C$.
We want the bound fixed-point variable to be discardable to model potentially terminating loops, but we do not want it to create a backdoor to duplicate linear resources.
We thus leverage the $\dwc \ucw A = {!A}$ encoding to allow the fixed point to be discardable and require fixed points to be created under a context $\Gamma_W, \Gamma_U$, consisting of only $W$ and $U$ assumptions, which are non-linear.
\[
  \infer[fix]{\ctp{\Gamma_W, \Gamma_U}{\fixc{f}{M}}{C}}
  {\ctp{\Gamma_W, \Gamma_U, f : \dwc \ucw C}{M}{C}}
\]
This is sufficient to model loops in hybrid quantum-classical algorithms, though we can support a more general fixed-point operation by introducing an affine runtime mode.

\subsection{Instrumenting a Quantum Program}
\label{ssec:instrumenting}
We have thus far described $C$ and $W$, which together model a classical control system that interacts with a quantum co-processor modeled by $Q$.
Now, it is natural to extend our metaprogramming language ($U$ and $L$) with the ability to generate $C$ code based on previously computed circuits, i.e., $Q$ code.
Although computational interpretations of adjoint logic typically take the upshift to represent capturing syntax, we instead use the downshift to capture $C$ code, since we want the captured $C$ code (and its embedded $Q$ code) \textbf{to be normalized}.
Thus, we require the pre-order $C \geq L$ so that $\dcl A_C$ is well-formed.

Finally, we want to actually run the generated $C$ code; we proceed by giving two distinct semantics for \PQD: a circuit-normalization semantics, defined for all modes, where values of type $\dcl A_C$ are ``fully normalized'' in the sense that it contains no subterms of modes $U$ and $L$, and a runtime semantics, defined only for normalized $C$-moded terms, that captures the execution of (normalized with respect to the circuit-normalization semantics) $C$-moded code and its interactions with the quantum co-processor.
The former ensures that all quantum circuits that are embedded in the generated $C$ code have been pre-computed.

\section{Operational Semantics and Metatheory}
We make precise our distinction between a ``circuit-normalization'' phase that aims to pre-compute all quantum circuits and a ``runtime'' phase that actually executes the underlying quantum computation.
We describe each semantics in turn.

\subsection{Circuit-Normalization Semantics}
\label{sec:circuit-normalization}
The circuit-normalization semantics for \PQD is given via a term-rewriting small-step semantics, which is heavily based on the semantics of Proto-Quipper-A~\cite{Kavanagh26ppdp}.
While reductions in the metaprogramming language $U/L$ are fairly standard, the semantics for the quantum language $Q$ and the classical runtime language $C$ are carefully designed to ensure that $C$-moded terms and any of its embedded $Q$-moded terms fully
eliminate all of their $U/L$ subterms.
This captures the intention of the circuit-normalization semantics as evaluating the metaprogramming logic in $U/L$ to generate a $C$-moded program that can readily be executed by a quantum system.

\subsubsection{Normal Forms}
\label{sec:circuit-normal-forms}
We let $N_m, M_m$ range over arbitrary terms in mode $m$, $V_m$ over normal forms, and $K_m$ over canonical forms and define each in turn.
We let $g$ range over gate constants, which we treat abstractly.
Since $Q$ and $C$ contain uninterpreted constants during the compilation phase (gate constants in $Q$ and measurement in $C$),
we also require neutral forms $R_Q$ and $R_C$ to capture terms that are blocked.
\begin{align*}
  V_m &\Coloneqq K_m \tag{for $m \in \{U, L\}$} \\
  V_Q &\Coloneqq K_Q \defor R_Q \defor \ltrivq{R_Q}{V_Q} \defor \lpairq{x}{y}{R_Q}{V_Q} \\
  V_W &\Coloneqq K_W \defor R_W \\
  V_C &\Coloneqq K_C \defor R_C \defor \ltriv{R_m}{V_C}{m} \defor \lpair{x}{y}{R_m}{V_C}{m} \tag{for $m \in \{Q, C\}$} \\
      &\defor \ldown{n}{C}{x}{R_n}{V_C} \defor \ifc{R_C}{V_C}{V_C} \tag{for $n \in \{Q, W\}$} \\
  K_m &\Coloneqq \trivm{m} \defor \pairm{V}{V}{m} \defor \lam{x}{N_m} \defor \susp{n}{m}{M_n} \defor \down{n}{m}{V_n} \tag{for $m \in \{U, L\}$} \\
  K_Q &\Coloneqq \trivq \defor \pairq{V_Q}{V_Q} \defor \lam{x}{V_Q} \\
  K_W &\Coloneqq \susp{C}{W}{M_C} \\
  K_C &\Coloneqq \trivm{C} \defor \pairm{V}{V}{C} \defor \lam{x}{V_C} \defor \down{Q}{C}{V_Q} \defor \down{W}{C}{V_W} \defor \fixc{x}{V_C} \\
  R_Q &\Coloneqq x \defor \app{g}{K_Q} \defor \app{g}{R_Q} \\
  R_W &\Coloneqq x \\
  R_C &\Coloneqq x \defor \app{R_C}{K_C} \defor \app{R_C}{R_C} \defor \meas{V_Q} \defor \force{C}{W}{R_W}
\end{align*}

Intuitively, while values in $U/L/W$ are standard, values in $Q$ and $C$ can begin with a chain of pattern matches (or branching) on neutral terms and end with canonical or neutral terms.
$W$ can also be neutral since we reduce under all binders for $C$ to normalize circuits, and
$C$ can eliminate $W$-moded terms via a $\ldown{W}{C}{x}{\ldots}{M_C}$ construct.

Canonicals capture terms arising from constructors, where we note that $W$ has been intentionally stripped of
all constructors except for $\mathsf{susp}$, since it is only used to model the exponential in $C$.
Neutrals capture terms that are blocked on uninterpreted constants, such as gates and measurement, and variables.
We specifically require $C$ and $Q$-moded functions to fully reduce their bodies, since they may contain $Q$-moded terms that we want to be fully reduced by our semantics.

We remark that $\meas{V_Q}$ is neutral as we treat measurement as an uninterpreted constant during this phase.
Unlike gate applications however, which require its argument to be canonical or neutral, we allow the body of a measurement
to be any normal form, i.e., it can contain a chain of pattern matches on neutral terms.
From a technical standpoint, we could have also required the body of the measurement to be neutral (not canonical since there are
no constructors of type $\qubit$).
However, we made an intentional design choice since commuting out any matches inside the body of a measurement can lead
down the line to multiple requests to the quantum co-processor, which is not desirable.

\subsubsection{Small-step Semantics}
The circuit-normalization semantics is largely based off of the semantics of Proto-Quipper-A and is given by a relation $\cst{M}{M'}$ on terms $M$ and $M'$.
As convention, we use metavariables introduced in the prior section to implicitly enforce that the reduction rules only apply to terms of the shapes (e.g., $V$ for values and $K$ for canonicals).
We begin by introducing the reduction rules for terms that are shared by all modes, consisting of congruence rules for all constructs except functions and $\mathsf{susp}$, which both block reductions in their bodies, and standard $\beta$ rules for each construct.

The following congruence rules are standard.
\[
  \infer{\cst{\app{M}{N}}{\app{M'}{N}}}
    {\cst{M}{M'}}
  \quad
  \infer{\cst{\app{K}{N}}{\app{K}{N'}}}
    {\cst{N}{N'}}
  \quad
  \infer{\cst{\ltriv{M}{N}{m}}{\ltriv{M'}{N}{m}}}
    {\cst{M}{M'}}
\]
\[
  \infer{\cst{\pairm{M}{N}{m}}{\pairm{M'}{N}{m}}}
    {\cst{M}{M'}}
  \quad
  \infer{\cst{\pairm{V}{N}{m}}{\pairm{V}{N'}{m}}}
    {\cst{N}{N'}}
  \quad
  \infer{\cst{\down{n}{m}{M}}{\down{n}{m}{M'}}}
    {\cst{M}{M'}}
\]
\[
  \infer{\cst{\lpair{x}{y}{M}{N}{m}}{\lpair{x}{y}{M'}{N}{m}}}
    {\cst{M}{M'}}
  \quad
  \infer{\cst{\ldown{n}{m}{x}{M}{N}}{\ldown{n}{m}{x}{M'}{N}}}
    {\cst{M}{M'}}
\]

The following congruence rules apply to both $Q$ and $C$, which enable reductions under functions and pattern matches blocked by neutrals. Since $C$ can match on $Q$ and $C$-moded units and pairs, these rules allow reductions in the bodies of these matches as well.
\[
  \infer{\cst{\ltriv{R}{N}{m}}{\ltriv{R}{N'}{m}}}
    {\cst{N}{N'}
    & (\text{where}\;m \in \{Q, C\})}
  \quad
  \infer{\cst{\lpair{x}{y}{R}{N}{m}}{\lpair{x}{y}{R}{N'}{m}}}
    {\cst{N}{N'}
    & (\text{where}\;m \in \{Q, C\})}
\]
\[
  \infer{\cst{\ldown{m}{C}{x}{R}{N}}{\ldown{m}{C}{x}{R}{N'}}}
    {\cst{N}{N'}
    & (\text{where}\;m \in \{Q, W\})}
  \quad
  \infer{\cst{\lam{x}{M_m}}{\lam{x}{M'_m}}}
    {\cst{M_m}{M'_m}
    & (\text{where}\;m \in \{Q, C\})}
\]

The following congruence rules only apply to $C$. In particular, we allow reductions under the fixed-point construct since they can contain quantum circuits that we want to be fully reduced, measurement, and bodies under conditionals that are blocked by neutrals.
\[
  \infer{\cst{\fixc{x}{M_C}}{\fixc{x}{M'_C}}}
    {\cst{M_C}{M'_C}}
  \quad
  \infer{\cst{\meas{M_Q}}{\meas{M'_Q}}}
    {\cst{M_Q}{M'_Q}}
\]
\[
  \infer{\cst{\ifc{R}{N_1}{N_2}}{\ifc{R}{N_1'}{N_2}}}
    {\cst{N_1}{N_1'}}
  \quad
  \infer{\cst{\ifc{R}{V}{N_2}}{\ifc{R}{V}{N_2'}}}
    {\cst{N_2}{N_2'}}
\]

And finally, the following $\beta$ rules are standard, where $[V/x]M$ denotes a capture-avoiding substitution of $V$ for $x$ in $M$, defined
in the usual way (see \Cref*{app:substitution}).
\[
  \infer{\cst{\app{(\lam{x}{M_m})}{V_m}}{[V_m/x]M_m}}{(m \in \{U, L\})}
  \quad
  \infer{\cst{\app{(\lam{x}{V_m})}{V_m}}{[V_m/x]V_C}}{(m \in \{Q, C\})}
\]
\[
  \cst{\force{m}{n}{\susp{m}{n}{M}}}{M}
  \quad
  \cst{\ltriv{\trivm{m}}{N}{m}}{N}
\]
\[
  \cst{\lpair{x}{y}{\pair{V_1}{V_2}{m}}{N}{m}}{[V_1/x, V_2/y]N}
  \quad
  \cst{\ldown{n}{m}{x}{\down{n}{m}{V}}{N}}{[V/x]N}
\]

Since $Q$ and $C$ contain uninterpreted constants that block reductions, we also include commuting conversion rules that allow neutral matches (and conditionals) to commute with eliminators.
We give two representative rules here involving matches on units only; we refer readers to \Cref*{app:compile-commuting} for the full listing of commuting conversion rules involving matches on units.
\[
  \cst{\app{(\ltriv{R}{V_1}{m})}{V_2}}{\ltriv{R}{(\app{V_1}{V_2})}{m}}
  \quad
  \cst{\app{V_1}{(\ltriv{R}{V_2}{m})}}{\ltriv{R}{(\app{V_1}{V_2})}{m}}
\]

\subsubsection{Properties of the Circuit-Normalization Semantics}
The circuit-normalization semantics enjoys type preservation and progress.
First, we state the usual substitution property, required to deal with cases involving substitution in the proof of type preservation.
\begin{lemma}[Substitution Property]
\label{lem:subst}
  If \ $\tp{\Gamma_1, x : A}{m}{M}{B}$ and $\tp{\Gamma_2}{m}{V}{A}$, then
  \\${\tp{\GammaIII}{m}{[V/x]M}{B}}$.
\end{lemma}

For our remaining metatheory, we must work on terms with free variables in $C$, $Q$, and $W$ since the circuit-normalization semantics performs reductions under these bindings.
Therefore, we use the notation $\Gamma_{WCQ}$ to denote a context that only contains $W/C/Q$ assumptions, including the empty context.
Note that typing guarantees that if a term has free variables in $Q$ or $C$, it must not be in $U$ or $W$, since $U \ngeq Q$, $U \ngeq C$, $W \ngeq Q$, and $W \ngeq C$.

\begin{theorem}[Type Preservation of Circuit-Normalization Semantics]
  If \ $\tp{\Gamma_{WCQ}}{m}{M}{A}$ and $\cst{M}{M'}$, then $\tp{\Gamma_{WCQ}}{m}{M'}{A}$.
\end{theorem}
\begin{proof}
  By induction on the derivation of $\cst{M}{M'}$. We give several representative cases in \Cref*{app:proofs}.
\end{proof}
We next state a normal forms lemma that characterizes the shape of normal forms in each mode, which is crucial for proving progress.
Since let-style eliminators and conditionals on neutral terms can appear in front of normal terms of any type in $Q$ and $C$, we use several shorthand notations.
First, we use the following shorthand to append matches and conditionals to a grammar defining $V_Q$ and $V_C$ respectively.
In the definition of $\mathsf{let}_C$, $m$ ranges over $Q$ and $C$ and $n$ ranges over $Q$ and $W$.
{
\begin{align*}
  (\mathsf{let}_Q) &\Coloneqq \ltrivq{R_Q}{V_Q} \defor \lpairq{x}{y}{R_Q}{V_Q} \\
  (\mathsf{let}_C) &\Coloneqq \ltriv{R_m}{V_C}{m} \defor \lpair{x}{y}{R_m}{V_C}{m} \defor \ldown{n}{C}{x}{R_C}{V_C} \\
                   &\defor \ifc{R_C}{V_C}{V_C} \defor \force{C}{W}{R_W}
\end{align*}
}
We write $R_C'$ to denote terms generated by $R_C$ that does not use $\mathsf{measure}$.

\begin{lemma}[Normal Forms]
\label{lem:compile-norm}
If \ $\tp{\Gamma_{WCQ}}{m}{V_m}{A_m}$ and $V_m$ is normal, then $V_m$ has the following shape depending on $A_m$. 

\begin{itemize}
  \item If $A_m = 1_m$, then $V_m = \trivm{m}$ for $m \in \{U, L, W\}$.
  \item If $A_Q = 1_Q$, then $V_Q = \trivq \defor R_Q \defor (\mathsf{let}_Q)$.
  \item If $A_C = 1_C$, then $V_C = \trivc \defor R_C' \defor (\mathsf{let}_C)$.
  \item If $A_m = A_1 \otimes_m A_2$, then $V_m = \pair{V_m}{V_m}{m}$ for $m \in \{U, L, W\}$.
  \item If $A_Q = A_1 \otimes_Q A_2$, then $V_Q = \pairq{V_Q}{V_Q} \defor R_Q \defor (\mathsf{let}_Q)$.
  \item If $A_C = A_1 \otimes_C A_2$, then $V_C = \pairc{V_C}{V_C} \defor R_C' \defor (\mathsf{let}_C)$.
  \item If $A_m = A_1 \multimap_m A_2$, then $V_m = \lam{x}{N_m}$ for $m \in \{U, L, W\}$.
  \item If $A_Q = A_1 \multimap_Q A_2$, then $V_Q = \lam{x}{N_Q} \defor g \defor (\mathsf{let}_Q)$.
  \item If $A_C = A_1 \multimap_C A_2$, then $V_C = \lam{x}{N_C} \defor R_C' \defor (\mathsf{let}_C)$.
  \item If $A_m = \uparrow^m_n A_n$, then $V_m = \susp{n}{m}{M_n}$ for $m \in \{U, L\}$.
  \item If $A_m = \uparrow^W_C A_C$, then $V_W = \susp{C}{W}{M_C} \defor x$.
  \item If $A_m = \downarrow^n_m A_n$, then $V_m = \down{n}{m}{V_n}$ for $m \in \{U, L, W\}$.
  \item If $A_C = \downarrow^Q_C A_Q$, then $V_C = \down{Q}{C}{V_Q} \defor R_C' \defor (\mathsf{let}_C)$.
  \item If $A_C = \dwc A_{W}$, then $V_C = \down{W}{C}{V_{W}} \defor \mkred{R_C} \defor (\mathsf{let}_C)$.
  \item If $A_Q = \qubit$, then $V_Q = R_Q \defor (\mathsf{let}_Q)$.
  \item If $A_C = \cbit$, then $V_C = R_C' \defor (\mathsf{let}_C)$.
\end{itemize}
\end{lemma}
Note the sole use of $\mkred{R_C}$ as opposed to $R_C'$, since $\dwc A_{W}$ is the only type in $C$ that can be generated by $\meas{V_Q}$, which has type $\dwc \ucw \cbit$.

\begin{theorem}[Progress of Circuit-Generation Semantics]
\label{thm:compile-progress}
If \ $\tp{\Gamma_{WCQ}}{m}{M}{A}$, then either $M$ is a value or there exists some $M'$ such that $\cst{M}{M'}$.
\end{theorem}
\begin{proof}
  By induction on the derivation of $\tp{\Gamma_{WCQ}}{m}{M}{A}$. We refer readers to \Cref*{app:proofs} for a sketch of the proof or
  the mechanized proof in our artifact.
\end{proof}
Finally, we obtain our desired property that the circuit-normalization semantics reduces $C$-code in a way such that any embedded $Q$-code do not contain any $U$ or $L$ subterms.
\begin{corollary}[Compilation Pre-Generates Circuits]
\label{cor:compilation-normalizes}
  If $\ltp{\cdot}{M}{{\dcl} A_C}$, then $M$ is some $\downcl{V}$ such that $V$ does not contain any subterms in modes $U$ or $L$ or there exists some $M'$ such that $\cst{M}{M'}$.
\end{corollary}

\subsection{Runtime Semantics}
\label{sec:runtime-semantics}
The runtime semantics captures the execution of normalized (from the circuit-normalization semantics) $C$-moded programs and their interactions with a quantum co-processor.
We therefore do not define any reduction rules for $U/L$ terms, since there are no $U/L$-moded subterms in a normalized $C$ term.

While circuit-normalization rewrote open terms to support normalizing circuits under binders, the runtime semantics is a deterministic call-by-value semantics with qubit references that operates on \emph{closed} terms.
Because quantum computation is inherently stateful, we pair these terms with a quantum runtime state \(\rho\) and a list \(\Sigma\) of qubit identifiers \(p, q, r, \ldots\) that uniquely identify qubits in \(\rho\).
Our treatment of quantum state is abstract, but \(\rho\) could be instantiated by a partial density matrix for simulation purposes.
Qubit identifiers are structural symbols that can be renamed, but they do not admit substitution.
They are embedded in terms using references \(\qref{q}\), which have type qubit.
These data form the configurations \(\config{\rho}{\Sigma}{M}\) rewritten by our small-step semantics.

\subsubsection{Normal forms}
Normal forms are significantly simpler at runtime.
Indeed, there are no neutral terms at runtime because we only execute closed programs, and gates now have computational meaning.
As a result, normal forms \(V_m\) at mode \(m\) correspond to typical functional values, where we add
the constants $0$ and $1$ to denote classical bits.
A configuration is normal if its term is normal:
\begin{align*}
  V_W &\Coloneqq \susp{C}{W}{M_C} \qquad
  V_C \Coloneqq 0 \defor 1 \defor \trivm{C} \defor \pairm{V}{V'}{C} \defor \lam{x}{M_C} \defor \down{Q}{C}{V_Q} \defor \down{W}{C}{V_W}\\
  V_Q &\Coloneqq \qref{q} \defor \trivm{Q} \defor \pairm{V}{V'}{Q} \qquad
  V_{\mathit{cfg}} \Coloneqq \config{\rho}{\Sigma}{V_m} \qquad (m \in \{\,Q, C, W\,\})
\end{align*}

\subsubsection{Small-step Semantics}
The runtime semantics for the \(C\) language has a standard call-by-value semantics with additional rules that
interact with the quantum co-processor, modeled by the runtime semantics for $Q$.
Measurement is operationally captured using a meta-level abstraction \(\MEAS{\rho}{q} = (\rho', b)\) that measures a qubit \(q\) in a state \(\rho\).
Measurement produces an updated state \(\rho'\) that does not refer to \(q\), along with the measured bit \(b \in \{\, 0, 1 \,\}\).
In the first rule, we capture the fact that measuring \(q\) collapses it by deleting it from the state \(\Sigma\).
\begin{gather*}
  \infer{
    \rst{
      \config{\rho}{\Sigma, q}{\meas{(\qref{q})}}
    }{
      \config{\rho'}{\Sigma}{b}
    }
  }{
    \MEAS{\rho}{q} = (\rho', b)
  }
  \qquad
  \infer{
    \rst{
      \config{\rho}{\Sigma}{\meas{M}}
    }{
      \config{\rho'}{\Sigma'}{\meas{M'}}
    }
  }{
    \rst{
      \config{\rho}{\Sigma}{M}
    }{
      \config{\rho'}{\Sigma'}{M'}
    }
  }
\end{gather*}
Bits are eliminated at runtime using conditionals
\begin{gather*}
  {
    \rst{
      \config{\rho}{\Sigma}{\ifc{1}{M_1}{M_0}}
    }{
      \config{\rho}{\Sigma}{M_1}
    }
  }
  \qquad
  {
    \rst{
      \config{\rho}{\Sigma}{\ifc{0}{M_1}{M_0}}
      }{
        \config{\rho}{\Sigma}{M_0}
      }
    }{
    }
    \\
    \infer{
      \rst{
        \config{\rho}{\Sigma}{\ifc{N}{M_1}{M_0}}
        }{
          \config{\rho'}{\Sigma'}{\ifc{N'}{M_1}{M_0}}
          }
        }{
          \rst{
            \config{\rho}{\Sigma}{N}
          }{
            \config{\rho'}{\Sigma'}{N'}
          }
        }
      \end{gather*}
And finally, the interaction between the layers is given by shift elimination:
\begin{gather*}
  \rst{
    \config{\rho}{\Sigma}{\ldown{m}{C}{x}{\down{m}{C}{V_m}}{N}}
  }{
    \config{\rho}{\Sigma}{\subst{V_m}{x}{N}}
  }
  \tag{
    $m \in \{Q, W\}$
  }
  \\
  {
    \rst{
      \config{\rho}{\Sigma}{\force{C}{W}{\susp{C}{W}{M}}}
    }{
      \config{\rho}{\Sigma}{M}
    }
  }
\end{gather*}
The remaining rules are standard.

The runtime semantics for the \(Q\)-layer executes circuits, which describe sequences of gate applications.
Gate applications are performed by the meta-level abstraction \(\gapp{g}{\rho}{V_Q} = (\rho'; V_Q')\).
It applies the gate \(g\) to the tuple of qubits in the state \(\rho\) identified by \(V_Q\); we know by a runtime normal forms lemma (analogous to \Cref{lem:compile-norm}, which we omit here) that \(V_Q\) must be a tensor of qubit references.
The result is an updated quantum state \(\rho'\), and a tuple \(V_Q'\) of references to globally fresh qubit identifiers identifying \(g\)'s outputs in \(\rho'\).
We write \(\qforget{V_Q}\) for the set of qubit identifiers appearing in \(V_Q\):
\[
  \qforget{\trivq} = {\cdot} \qquad \qforget{\qref{q}} = q \qquad \qforget{\pair{V}{V'}{Q}} = \qforget{V}, \qforget{V'}
\]
Gate application is then given by:
\begin{gather*}
  \infer{
    \rst{
      \config{\rho}{\Sigma, \qforget{V_Q}}{\app{g}{V_Q}}
    }{
      \config{\rho}{\Sigma, \qforget{V_Q'}}{V_Q'}
    }
  }{
    \gapp{g}{\rho}{V_Q} = (\rho'; V_Q')
  }
\end{gather*}
As a special case, qubit initialization is given by the \(\gsty{init} : 1 \multimap \qubit\) gate.
Since \(\mathsf{apply}\) is assumed to generate globally fresh identifiers, there is no risk of collision between qubit identifiers:
\begin{gather*}
  \infer{
    \rst{
      \config{\rho}{\Sigma}{\app{\gsty{init}}{\trivq}}
    }{
      \config{\rho'}{\Sigma, q}{\qref{q}}
    }
  }{
    \gapp{\gsty{init}}{\rho}{\cdot} = (\rho'; \qref{q})
  }
\end{gather*}
The remaining computation rules are standard for deterministic call-by-value semantics:
\begin{gather*}
  \infer{
    \rst{
      \config{\rho}{\Sigma}{\app{M}{N}}
    }{
      \config{\rho'}{\Sigma'}{\app{M'}{N}}
    }
  }{
    \rst{
      \config{\rho}{\Sigma}{M}
    }{
      \config{\rho'}{\Sigma'}{M'}
    }
  }
  \quad
  \infer{
    \rst{
      \config{\rho}{\Sigma}{\app{V}{N}}
    }{
      \config{\rho'}{\Sigma'}{\app{V}{N'}}
    }
  }{
    \rst{
      \config{\rho}{\Sigma}{N}
    }{
      \config{\rho'}{\Sigma'}{N'}
    }
  }
  \\
  {
  \small
  {
    \rst{
      \config{\rho}{\Sigma}{\ltriv{\trivq}{M}{m}}
    }{
      \config{\rho}{\Sigma}{M}
    }
  }
  \quad
  \infer{
    \rst{
      \config{\rho}{\Sigma}{\ltriv{M}{N}{m}}
    }{
      \config{\rho'}{\Sigma'}{\ltriv{M'}{N}{m}}
    }
  }{
    \rst{
      \config{\rho}{\Sigma}{M}
    }{
      \config{\rho'}{\Sigma'}{M'}
    }
  }
}
  \\
  {
    \rst{
      \config{\rho}{\Sigma}{\lpair{x}{x'}{\pair{V}{V'}{m}}{M}{m}}
    }{
      \config{\rho}{\Sigma}{\subst{V, V'}{x, x'}{M}}
    }
  }
  \\
  \infer{
    \rst{
      \config{\rho}{\Sigma}{\lpair{x}{x'}{M}{N}{m}}
    }{
      \config{\rho'}{\Sigma'}{\lpair{x}{x'}{M'}{N}{m}}
    }
  }{
    \rst{
      \config{\rho}{\Sigma}{M}
    }{
      \config{\rho'}{\Sigma'}{M'}
    }
  }
\end{gather*}

\subsubsection{Safety properties}

We prove runtime type safety for configurations.
To type configurations, we begin by making our hypothetical typing judgment \(\tp{\Delta}{m}{M}{A}\) \textit{parametric} to support qubit references.
The resulting judgment, \(\ptp{\Sigma}{\Delta}{m}{M}{A}\), means that the program \(M\) is well-typed in the presence of qubit identifiers \(\Sigma\).
It is inductively defined by extending the typing judgments (for $W, C, $ and $Q$ modes) given in \Cref{sec:statics} with a linear parameter list \(\Sigma\) and the following new rule typing references:
\[
  \infer{
    \ptp{q}{\cdot}{Q}{\qref{q}}{\qubit}
  }{
  }
\]
The remaining rules are extended analogously, for instance:
\begin{gather*}
  \infer{
    \ptp{\Sigma_1, \Sigma_2}{\Delta_1, \Delta_2}{Q}{\pairq{M}{N}}{A_1 \otimes A_2}
  }{
    \ptp{\Sigma_1}{\Delta_1}{Q}{M}{A_1}
    &
    \ptp{\Sigma_2}{\Delta_2}{Q}{N}{A_2}
  }
  \quad
  \infer{
    \ptp{\Sigma}{\Delta}{C}{\lam{x}{M}}{A \multimap B}
  }{
    \ptp{\Sigma}{\Delta, x : A}{C}{M}{B}
  }
\end{gather*}

To state and prove our runtime safety properties in the presence of abstract meta-operations, we must specify when an abstract quantum state \(\rho\) is consistent with a set of identifiers \(\Sigma\).
They are consistent, written \(\consistent{\rho}{\Sigma}\), if:
\begin{itemize}
\item for all \(q \in \Sigma\), \(\MEAS{\rho}{q}\) is defined;
\item for all \(\Sigma' \subseteq \Sigma\) and gates \(g : \qarr{U}{S}\), if \(\ptp{\Sigma'}{\cdot}{Q}{V_Q}{U}\), then \(\gapp{g}{\rho}{V_Q} = (\rho'; V_Q')\) is defined and \(\ptp{\qforget{V_Q'}}{\cdot}{Q}{V_Q'}{S}\).
\end{itemize}
Intuitively, these criteria ensure that measurement and gate application are always possible and type-preserving at runtime.

A configuration is well-typed whenever its term is well-typed relative to the set of current identifiers, and its state and identifiers are consistent.
\[
  \infer{
    \cfgtp{\Delta}{\config{\rho}{\Sigma}{M}}{A_m}
  }{
    \ptp{\Sigma}{\Delta}{m}{M}{A_m}
    &
    \consistent{\rho}{\Sigma}
  }
\]

\begin{theorem}[Runtime Preservation]
  \label{prop:dapq:1}
  If \(\cfgtp{\cdot}{\config{\rho}{\Sigma}{M}}{A}\) and \(\rst{\config{\rho}{\Sigma}{M}}{\config{\rho}{\Sigma'}{M'}}\), then \(\cfgtp{\cdot}{\config{\rho}{\Sigma'}{M'}}{A}\).
\end{theorem}

\begin{theorem}[Runtime Progress]
  \label{prop:dapq:2}
  If \(\cfgtp{\cdot}{\config{\rho}{\Sigma}{M}}{A}\), then \(\rst{\config{\rho}{\Sigma}{M}}{\config{\rho}{\Sigma'}{M'}}\) for some \(\Sigma'\) and \(M'\), or \(M\) is a value.
\end{theorem}
We omit both proofs here, since they are simpler than the corresponding theorems for the circuit-normalization semantics.

\section{Mechanization}
\lstset{language=Beluga}
We mechanized a superset of \PQD in the proof assistant Beluga~\cite{Pientka10ijcar} based on the mechanization of Proto-Quipper-A~\cite{Kavanagh26ppdp}.
Mechanizations in Beluga typically proceed by encoding the syntax and semantics of the object language,
in our case, (a superset of) \PQD, in the logical framework LF~\cite{Harper93acm}.
Theorems about the object language are then proven by (recursive), total functions over these LF objects using Beluga's
reasoning framework.
To keep the presentation concise, we focus only on the LF encoding of the syntax and semantics of \PQD, especially
focusing on where our mechanization differs from our developments in this paper.
We refer the reader to our artifact for further details, including the mechanization of the metatheory of \PQD using
Beluga's reasoning framework.

\subsection{Modes and Types}
We represent the modes of \PQD as an LF type with a constructor for each mode:
\begin{lstlisting}
mode : type.  m/u : mode. m/q : mode. ...
\end{lstlisting}
We encode types to be parametrized by their mode; for instance, the type \lstinline|ℚ|, corresponding to  $\qubit$, is a type with mode \lstinline|m/q|.
\begin{lstlisting}
tp : mode → type.
ℚ : tp m/q.
⊸ : tp K → tp K → tp K.
\end{lstlisting}
Here, we make a simplifying choice to allow arbitrary functions for all modes; in particular, we do not restrict $Q$ functions in our mechanization.

\subsection{Terms}
We parameterize our terms with their mode and type, which enables us to give a single definition for constructors for all modes.
\begin{lstlisting}
tm : {K:mode} tp K → type.
app : tm K (⊸ A B) → tm K A → tm K B.
pair : tm K A → tm K B → tm K (⊗ A B).
\end{lstlisting}
This definition states that a term is parameterized by its mode and type (with said mode); while the mode information is redundant since types already contain mode information, lifting the mode at this level simplifies the mechanization.
We use LF's support for higher-order abstract syntax to represent the binding structure of \PQD's terms, which enables us to avoid encoding variables and binders explicitly.
For instance, function abstraction is represented through LF's function space:
\begin{lstlisting}
lam : (tm K A → tm K B) → tm K (⊸ A B).
\end{lstlisting}

We also use a generic match construct to handle our three let-style eliminators by defining
a generic \lstinline|pat| type for patterns that match on some term to produce a new term.
\begin{lstlisting}
pat : {K:mode} tp K → {K':mode} tp K' → type.
match : tm K A → pat K A K' B → tm K' B.
pat/pair : (tm K A → tm K B → tm K' C) → pat K (⊗ A B) K' C. % letp
\end{lstlisting}

\subsection{Mode Safety Predicates}
Since we leverage LF's (intuitionistic) function space to represent all binders of \PQD, we must additionally ensure
that all bound variables respect mode constraints.
In particular, we must ensure that bound variables respect their substructural constraints, i.e., $Q$, $C$, and $L$
variables must be used linearly, and that they do not cause any dependencies that violate our mode ordering.

We extend a technique by \citet{Crary10icfp} to define a local mode safety predicate over LF's function space:
\begin{lstlisting}
msf : (tm Ka A -> tm Kb B) -> type.
\end{lstlisting}
This predicate ``localizes'' mode safety by capturing LF functions that use its input term in a way that
that respect our mode conditions.
In particular, this predicate ensures two invariants: (1) the input term is used in a way that respects its substructural constraints, and (2) the input term is not used in a way that violates our mode ordering.

For instance, LF's identify function that maps terms to identical terms is trivially mode safe for all modes,
since terms of a given mode can depend on itself, and it does not violate any substructural constraints of any mode.
\begin{lstlisting}
msf/var : msf (\x.x).
\end{lstlisting}

However, this requires the assumption to be used, which is too restrictive for checking mode safety of $U$ and $W$ assumptions,
which allow weakening.
We therefore axiomatically state that any function that takes as input a $U$-moded term is mode safe
if it is being used in $U, L, Q$, and $C$ modes (but not $W$, since $U \ngeq W$).
\begin{lstlisting}
msf/var/UU : {f:tm m/u A -> tm m/u B} msf f.
msf/var/UL : {f:tm m/u A -> tm m/l B} msf f.
msf/var/UQ : {f:tm m/u A -> tm m/q B} msf f.
msf/var/UC : {f:tm m/u A -> tm m/c B} msf f.
% similarly for W (where W can appear in W, Q, and C)
\end{lstlisting}

We leverage higher-order unification to concisely encode the lack of contraction in $L/Q/C$, i.e., that
assumptions of these modes cannot be duplicated.
For instance, the following rule enforces that the argument \lstinline|x| does not appear
in \lstinline|N| and that it must only appear in \lstinline|(M x)|, which captures
a context split.
\begin{lstlisting}
msf/app/1 : msf M -> msf (\x. app (M x) N).
\end{lstlisting}
It is worth mentioning that we do not need to explicitly check if \lstinline|x| is of mode $L$, $Q$, or $C$ since
if it is of mode $U$ or $W$, then we can immediately use the axioms defined above to conclude that the function is mode safe.

These mode safety predicates are importantly \emph{local}, in the sense that it only checks the mode safety of a single
variable.
Therefore, to speak of the mode safety of an entire term, we define a \lstinline|bmsf| predicate that ensures that
all bindings within a term satisfy the \lstinline|msf| predicate, which in turn ensures that any closed term that satisfies
\lstinline|bmsf| is properly mode safe.
Its definition recursively traverses terms and checks that all bindings satisfy the \lstinline|msf| predicate:
\begin{lstlisting}
bmsf : tm K A → type.
% app : there are no bindings, so we only traverse recursively
bmsf/app : bmsf M → bmsf N → bmsf (app M N).
% lam : has a binding, so we check that the treats its input in a mode safe way
% and also recursively traverse the body
bmsf/lam : msf (M : tm K A → tm K B)
	   → ({x:tm K A}bmsf x → bmsf (M x))
	   → bmsf (lam M).
\end{lstlisting}
Thus, we have the following correspondence: our on-paper typing derivation $\tp{\cdot}{m}{M}{A}$ corresponds to:
(1) The intrinsically-typed term \lstinline|M: tm m A|; (2) a proof that all internal bindings in $M$ respect mode safety \lstinline|bmsf M|.
Open terms informally correspond to a proof that the term is mode safe with respect to each of its free variables, which is captured by the \lstinline|msf| predicate.

To further increase confidence in our encoding, we mechanize a proof of the \emph{independence principle}, which states that if \lstinline|msf(\x. P)|, then the mode of \lstinline|x| is greater than or equal to the mode of \lstinline|(P x)|.

\subsection{Type-Preserving Reductions}
We encode both circuit-normalization and runtime semantics as relations on terms.
Since our terms are intrinsically-typed, we obtain type preservation (but not preservation of mode safety) for free, since all reductions are, by construction, type-preserving.
We do not mechanize the proofs of preservation of mode safety, because we believe it follows straightforwardly from Kavanagh et al.'s mechanization of Proto-Quipper-A, which does prove that reductions preserve mode safety.

\paragraph{Circuit Normalization}
We define normal forms, canonical forms, neutrals, and the circuit-normalization reductions following the definitions given in \Cref{sec:circuit-normalization}, with some differences due to our mechanization being a superset of \PQD.
\begin{lstlisting}
norm: tm K A → type.  can : tm K A → type.  neu : tm K A → type.
↦   : tm K A → tm K A → type.
\end{lstlisting}

\paragraph{Runtime Semantics}
We define normal forms (which we call values in our mechanization) for the runtime semantics and the runtime reduction relation following the definitions given in \Cref{sec:runtime-semantics}.
Here, instead of tracking qubit identifiers in a configuration, we track them directly in the term by introducing a
constructor \lstinline|exq : (qn → tm K A) → tm K A| that binds a fresh qubit identifier \lstinline|qn| in its body.
We found this formulation more convenient to mechanize, since it avoids the need to track qubit identifiers separately.
\begin{lstlisting}
val : tm K A → type.
⇛ : tm K A → tm K A → type.
\end{lstlisting}

Finally, we mechanize the progress proof for both circuit-normalization and runtime semantics.
In both cases, we use an auxiliary judgment to capture the fact that a term is either normal or can step; we give
the definition for the circuit-normalization semantics, but the runtime semantics is analogous.
\begin{lstlisting}
↦_n : tm K A → type.   ↦_n/n : norm V → ↦_n V.   ↦_n/↦ : ↦ P Q → ↦_n P.
\end{lstlisting}
The progress statement for the circuit-normalization semantics, in words, states that for any term \lstinline|M : tm K A|,
we can derive \lstinline|↦_n M|, i.e., \lstinline|M| is either normal or can step.
Again, the statement for the runtime semantics is analogous.
We emphasize that neither of the progress proofs depend on \lstinline|P| being mode safe at all, which is not surprising since
the mode safety predicates only serve to add restrictions to terms.

\lstset{language=dapq}

\section{Extending Each Layer}
\label{sec:extending-F}
An important consequence of the modular design of \PQD is that all languages are easily extensible without affecting the design of the other languages.
In our case, this is mostly relevant for the metaprogramming language, since we envision no restrictions on the system that this would run on, though we do discuss several extensions to both $Q$ and $C$.

\subsection{Extending the Metaprogramming Language}
We first sketch several desirable extensions to the $U$ fragment of the metaprogramming language.
We can repeat the argument to similarly extend $L$, which we omit here.
The general recipe to extend $U$ with additional constructs is to extend the statics with new types and only the circuit-normalization semantics -- the runtime semantics is unaffected since $U$ terms do not appear in normalized $C$ or $Q$ terms, though we must ensure that our circuit-normalization semantics do fully reduce the additional constructs if they have type $\uql A_L$ for some $A_L$.

\subsubsection{Recursive Types and Recursion}
To enable code generation for families of quantum circuits, e.g., to implement \Cref{ex:fold}, we extend the metaprogramming language with lists alongside its recursor.
We use $\mathsf{list}_A$ to denote a list type for some fixed type $A \in \{U\}$; supporting polymorphism is orthogonal, although
we do not see any technical challenges in doing so.

The typing rules are straightforward and given below, where $\mathsf{nil}$ is the empty list, $\mathsf{cons}(M, N)$ constructs a list with head $M$ and tail $N$, and $\mathsf{rec}(M, x. xs. r. N)$ is the recursor that takes a list $M$, 
the base case $N_1$, and the recursive case $N_2$ that takes the head $x$, tail $xs$, and recursive result $r$.
\[
  \small
  \infer[nil]{\utp{\Gamma}{\mathsf{nil}}{\mathsf{list}_A}}{}
  \quad
  \infer[cons]{\utp{\Gamma}{\mathsf{cons}(M, N)}{\mathsf{list}_A}}{
  \utp{\Gamma}{M}{A} & \utp{\Gamma}{N}{\mathsf{list}_A}}
\]
\[
  \small
  \infer[rec]{\utp{\Gamma}{\mathsf{rec}(M, N_1, x. xs. r. N_2)}{B}}{
  \utp{\Gamma}{M}{\mathsf{list}_A} &
  \utp{\Gamma}{N_1}{B} &
  \utp{\Gamma, x : A, xs : \mathsf{list}_A, r : B}{N_2}{B}}
\]
Since $U$ terms can only depend on unrestricted $U$ assumptions, we omit any context merge operations in the above rules for ease of presentation.

The circuit-normalization semantics correspond exactly to a typical small-step semantics for lists.
Importantly, we ensure that the recursor is fully reduced in the circuit-normalization semantics, i.e., $\mathsf{rec}(-)$ does not appear as a normal form, preserving the property that normalized $C$ terms do not contain any metaprogramming language subterms.

\subsubsection{Code Analysis on $Q$ and $C$}
Our language-level separation of $Q$ and $C$ from the metaprogramming language $U$ and $L$, alongside our circuit-normalization semantics, in principle allows us to
support pattern matching on $Q$ and $C$ syntax in the metaprogramming language.
This, combined with the ability for the metaprogramming language to recursively traverse into $Q$ and $C$ sub-terms, gives a
full metaprogramming system that can, for instance, optimize and analyze hybrid programs.
To that end, we leverage our circuit-normalization semantics, which limits the shape of $Q$ and $C$ terms, due to
\Cref{cor:compilation-normalizes}.

Yet, we require significant work to extend \PQD with a type-safe pattern matching on $Q$ and $C$ terms.
The major challenge is that we must traverse terms that may contain binders, meaning we must support matching on open terms,
and thereby introduce first class contexts.
Since both $Q$ and $C$ are linear languages, we must also ensure that these contexts satisfy linearity constraints upon traversing
terms that cause context splits, which is non-trivial.
Nevertheless, both LINCX~\cite{Georges17esop} and FuSes~\cite{Sano25icfp} have developed type-safe pattern matching on linear terms,
and we expect \PQD to be able to leverage their techniques, though with some simplifications due to our normalization property.

In particular, we believe the technique introduced in FuSes applies readily to our setting, which amounts to \textit{contextualizing} ${\uql} A_Q$ and ${\dcl} A_C$ in the style of \citet{Nanevski08tocl}.
In this extension, ${\uql}(\Delta_Q \vdash A_Q)$ denotes open $Q$ terms of type $A_Q$ with free variables in $\Delta_Q$, which only consists of $Q$-level assumptions, and similarly, ${\dcl}(\Delta_C \vdash A_C)$ denotes open $C$ terms of type $A_C$ with free variables in $\Delta_C$.
We extend both $\mathsf{susp}_{Q}^{L}$ and $\mathsf{down}_{L}^{C}$ to take as input a corresponding context, which act as binders for the free variables inside both constructs.
We further extend $L$ with first-class contexts of $Q$ and $C$ assumptions, first-class $Q$ and $C$ types, and a limited form of dependent types that allows for the contextualized types to depend on these first-class contexts and types.

\subsubsection{Dependent Types}
Another natural extension is to add dependent types to the metaprogramming language, which would enable describing families of circuits parametrized by metaprogramming language values.
Of particular interest is the ability to describe families of circuits parameterized by natural numbers, for instance, a function that generates an $n$-ary circuit depending on a natural number $n$.
This is relevant, for instance, to generate families of quantum circuits of a specific size.

\subsection{Extending $C/Q$}

The modular design of \PQD makes it also straightforward to extend the $C$ and $Q$ modes with additional constructs.

The first kind of extensions that we discuss are extentions to the runtime layer, $C$.
In this paper we presented a minimal instance of $C$, with just bits and basic control flow. We did this intentionally, in order to capture the idea that the runtime classical machine has limited capabilities and/or time in which to execute, as well as to simplify our presentation. However, it is possible to extend $C$ in a variety of ways. 
For instance, we can add additional bitwise operations, such as constants $0$ and $1$ and binary operators like boolean \emph{and} and \emph{or}. We could also add other classical datatypes such as integers, floating point numbers, or lists.
In general, since $C$ is a classical functional language, we do not see any technical challenges in
extending $C$ with common programming features and data structures.

The second kind of extensions are to the quantum runtime layer, $Q$.
Our development of $Q$ treated gates as uninterpreted functional symbols, thereby allowing for a general framework for
describing quantum circuits under any set of gates.
We could extend $Q$ to also support \emph{parametrized gates}, which take as input some classical ($C$-moded) parameter. The typical examples are the Pauli rotation gates $\gsty{RX}(\theta)$, $\gsty{RY}(\theta)$, and $\gsty{RZ}(\theta)$, where $\theta$ is a classical floating point number representing the rotation angle.
Rotation gates are ubiquitous in variational algorithms such as QAOA and other NISQ optimization algorithms~\cite{cerezo2021variational}.

We can support these gates by extending $C$ with a type to represent these angles, e.g., a \lstinline|float| type, and
extending $Q$ with gates that depend on these angles:
\[
  \infer{
    \qtp{\Gamma}{\gsty{RX}_M}{\qubit \multimap_Q \qubit}
  }{
    \ctp{\Gamma}{M}{\mathsf{float}}
  }
  \quad
  \infer{
    \qtp{\Gamma}{\gsty{RY}_M}{\qubit \multimap_Q \qubit}
  }{
    \ctp{\Gamma}{M}{\mathsf{float}}
  }
  \quad
  \infer{
    \qtp{\Gamma}{\gsty{RZ}_M}{\qubit \multimap_Q \qubit}
  }{
    \ctp{\Gamma}{M}{\mathsf{float}}
  }
\]

\section{Related Work}
\label{sec:rel-work}

Quipper and the Proto-Quipper family of languages~\cite{GLRSV2013-pldi,RS2017-pqmodel,FKS2020-lindep,Fu23popl} are a series of programming languages that provide high-level abstractions for programming quantum circuit generation. One of its key insights is the separation of circuit generation time from circuit runtime, which \PQD is highly inspired by. Several of these variants are especially related to \PQD, including Proto-Quipper-A and three varieties that support dynamic lifting.

\paragraph{Proto-Quipper-A}
\PQD builds on Kavanagh et al's Proto-Quipper-A~\cite{Kavanagh26ppdp}, which separates the language into modes based on adjoint type systems. Proto-Quipper-A introduces the $U$, $L$, and $Q$ modes and their associated type systems and a circuit-normalization semantics that mostly corresponds to the circuit-normalization semantics of \PQD for our $U$, $L$, and $Q$ modes.
However, Proto-Quipper-A only considers quantum circuit generation and does not consider execution of the generated circuits nor a classical runtime language to support dynamic lifting.
\PQD's $C$ mode (classical runtime) and its associated statics and dynamics, alongside the runtime semantics, can be viewed as an extension of Proto-Quipper-A.
They mechanize Proto-Quipper-A in Beluga with a similar approach to our mechanization of \PQD, though our additional
language features, including $C/W$ modes and the two semantics alongside their progress proofs, constituted
significant work.

More technically, \PQD omits Proto-Quipper-A's ability for $L$-moded terms to match on $Q$-moded terms, which enables the encoding of Quipper-style combinators \texttt{box} and \texttt{unbox}.
These combinators enable functions in $L$ of type ${\uql}U \multimap_L {\uql}S$ to be directly mapped to circuits ${\uql}(U \multimap_Q S)$, which encodes the $\text{Circ}(U, S)$ type in Proto-Quipper.
We chose to simplify the design of \PQD by omitting these combinators and programming circuits directly in $Q$, which, combined with our circuit-normalization semantics, corresponds exactly to their circuit generation.

We expect that \texttt{box} and \texttt{unbox} could be easily incorporated back into \PQD by re-introducing the ability for $L$-moded terms to match on $Q$-moded terms.
Matching on $Q$-moded terms would be possible to add to \PQD in the same way as we added matching on $Q$-moded terms in the $C$ layer in \Cref{ssec:statics-classical}.




Proto-Quipper-Dyn \cite{Fu23popl} models a view of dynamic lifting where some circuit generation can occur at quantum runtime when the results depend on measurement results. Like \PQD, Proto-Quipper-Dyn has a separate circuit generation semantics and circuit execution semantics, but unlike \PQD, their circuit execution semantics can invoke the circuit generation semantics in the process of dynamic lifting. Proto-Quipper-Dyn works to limit the amount of circuit generation that occurs at runtime by tracking dependencies via a type system, so that circuits or subcircuits without dependencies on dynamic lifting can be boxed at circuit generation time.
In contrast, our language-level separation most importantly enables all quantum circuits to be pre-generated, which is a powerful and desirable feature that is hard to achieve in their setting where circuit-normalization and runtime classical logic are intertwined.

Proto-Quipper-L~\cite{lee2021concrete} takes a slightly different approach, where instead of generating ordinary circuits, they generate a generalized type of circuit (which they refer to as a quantum channel)---a tree structure that includes branching as a result of measurement. For example, instead of a \PQD-style hybrid program such as
\lstinline|if (measure q) then M else N|, a corresponding Proto-Quipper-L program would have the form
\lstinline|measure q M N|. Like \PQD, Proto-Quipper-L has the effect of cleanly separating compilation time from runtime, and no circuit generation is done at runtime. 
However, a negative consequence is that every instance of dynamic lifting doubles the size of this tree structure, which comes at the cost of exponential blowup of the generalized circuits.
Consider the following example in \PQD, where the $C$ layer is extended with the \lstinline|&| operator on bits:
\begin{lstlisting}
if (measure q1) & ... & (measure qn) then M else N
\end{lstlisting}
\PQD executes this function directly in C by performing the conditional \lstinline{b1 & ... & bn} at runtime. The compiled binary only needs a single representation of \lstinline|M| and \lstinline|N|.
In contrast, Proto-Quipper-L must duplicate \lstinline{N} $n-1$ times:
\begin{lstlisting}
  measure q1 N (measure q2 N (... (measure qn N M)))
\end{lstlisting}

Proto-Quipper-K~\cite{colladan2023dynamic} draws on type-and-effect systems to support dynamic lifting. Like Proto-Quipper-L, Proto-Quipper-K produces generalized quantum circuits with a branching structure. In their version of generalized circuits, gates can be applied conditionally based on the classical value of a measurement result. This acturally reflects the capabilities of hardware devices like IBM that support limited on-chip classical control. This representation fixes the scalability issue of Proto-Quipper-L as the example above can be expressed compactly as follows:
\begin{lstlisting}
  measure(q1) -> x1; ... ; measure(qn) -> xn; lift(x1) => b1; ... ; lift(xn) => bn;
  (b1 & ... & bn) ? M  ~(b1 & ... & bn) ? N
\end{lstlisting}
While these generalized circuits are quite expressive, they are still limited to conditional circuit applications. In contrast, \PQD is not limited to a particular model of classical control. While we focus on bits and conditionals in the classical runtime layer, we can easily support more complex classical operations such as runtime function calls or even external library calls.

\paragraph{QWIRE} Another quantum circuit language that models dynamic lifting is QWIRE~\cite{Paykin17popl}, whose linear/non-linear type system is closely related to \PQD's adjoint type system. QWIRE circuits are not normalized under the dynamic lifting operator, which means that after lifting, circuit generation must occur at runtime. 
QWIRE is also relevant in that it enables optimization through pattern-matching of circuits, whereby circuits can be treated as data structures that could be manipulated by the meta-language.

\paragraph{Recursion} \PQD supports recursion at runtime, which is crucial to program repeat-until-success circuits, but it is not
obvious how to extxend Proto-Quipper-L or Proto-Quipper-K with runtime recursion.
The authors of Proto-Quipper-Dyn mention that their system can be extended to support recursion, but it has not been formalized \cite[Section 6]{Fu23popl}.

\paragraph{Adaptive circuits and dynamic lifting}

In recent years, quantum hardware developers have increasingly added support for \emph{adaptive circuits}, in which the QPU itself adapts the control flow based on dynamic measurement results~\cite{carrera2024dynamic, ransford2025helios98qubittrappedionquantum}. IRs that support adaptive circuits include QIR~\cite{lubinski2022advancing}, OpenQASM 3.0~\cite{cross2022openqasm}, and HUGR~\cite{koch2025hugr}.

We argue that \PQD's syntactic separation into classical and quantum modes is appropriate for hardware that supports adaptive circuits. Even with adaptive quantum control on-chip, there are still reasons to syntactically separate classical from quantum operations. For instance, grouping quantum operations together allows compilers to more easily apply quantum circuit optimization techniques, while grouping classical operations together allows compilers to take advantage of the wide array of classical optimization infrastructure. 

Second, even as hardware support for adaptive circuits grows, the on-chip programming model is often still limited. For example, IBM machines targeting OpenQASM 2.0 allow gates to be conditioned on individual measurement result bits, but do not support more general classical computation~\cite{cross2017open}. In such a setting, the $Q$ mode could be extended to capture this restricted form of on-chip classical control, while the $C$ mode handles richer computations delegated to the full classical control system.
The related technique of circuit cutting~\cite{peng2020simulating} partitions a circuit into subcircuits executed independently to work around qubit limitations; the hybrid structure of \PQD naturally accommodates the orchestration logic required to coordinate subcircuit execution and combine the results.

\section{Conclusion}
We present \PQD, a language designed specifically for programming hybrid quantum-classical systems.
We provide a system that separates our two stages of circuit-normalization and runtime, which we see as crucial to
realistically implement hybrid algorithms; time spent at runtime is time spent with the quantum co-processor active, which is orders of magnitude more expensive and should be minimized.
We give a formalization of \PQD's statics and importantly, its two dynamics, which distinguish between circuit-normalization and runtime execution of hybrid programs.
We prove preservation and progress for both semantics, and also mechanize those proofs in the Beluga proof assistant.

\PQD's use of adjoint logic makes it modular and robust, and we see many interesting future directions, which we have sketched in \Cref{sec:extending-F}.

\newpage

\bibliography{bibliography.bib}

\newpage
\appendix

\section{Additional Definitions}

\subsection{Context Merge}
\label{app:merge}
\begin{align*}
  \Gamma_1, x:A_m \bowtie \Gamma_2 &= (\Gamma_1 \bowtie \Gamma_2), x:A_m
    \tag{for $m \in \{L, Q, C\}$}
  \\
  \Gamma_1 \bowtie \Gamma_2, x:A_m &= (\Gamma_1 \bowtie \Gamma_2), x:A_m
    \tag{for $m \in \{L, Q, C\}$}
  \\
  \Gamma_1, x:A_m \bowtie \Gamma_2, x:A_m &= (\Gamma_1 \bowtie \Gamma_2), x:A_m
    \tag{for $m \in \{U, W\}$}
\end{align*}

\subsection{Substitution}
\label{app:substitution}
\[
  [V/x]x = V
  \quad
  [V/x]y = y
  \quad
  [V/x](\app{M}{N}) = \app{[V/x]M}{[V/x]N}
  \quad
  [V/x](\lam{y}{M}) = \lam{y}{[V/x]M}
\]
\[
  [V/x](\trivm{m}) = \trivm{m}
  \quad
  [V/x]\pairm{M}{N}{m} = \pairm{[V/x]M}{[V/x]N}{m}
\]
\[
  [V/x](\susp{m}{n}{M}) = \susp{m}{n}{[V/x]M}
  \quad
  [V/x](\down{n}{m}{M}) = \down{n}{m}{[V/x]M}
\]
\[
  [V/x](\ltriv{M}{N}{m}) = \ltriv{[V/x]M}{[V/x]N}{m}
\]
\[
  [V/x](\lpair{y}{z}{M}{N}{m}) = \lpair{y}{z}{[V/x]M}{[V/x]N}{m}
\]
\[
  [V/x](\ldown{m}{n}{y}{M}{N}) = \ldown{m}{n}{y}{[V/x]M}{[V/x]N}
\]
\[
  [V/x](\ifc{M}{N_1}{N_2}) = \ifc{[V/x]M}{[V/x]N_1}{[V/x]N_2}
\]
\[
  [V/x](\meas{M}) = \meas{([V/x]M)}
\]

\subsection{Commuting Conversions for Circuit-Normalization Semantics}
\label{app:compile-commuting}
Note that $C$-terms can match on $Q$-terms, so the commuting conversions below involving matches
on unit and pairs are overloading three cases: $C$-terms matching on $Q$-terms, $C$-terms matching on $C$-terms, and $Q$-terms matching on $Q$-terms.
\[
  \infer{\cst{\app{(\ltriv{R}{V}{m})}{N}}{\ltriv{R}{(\app{V}{N})}{m}}}
    {}
  \quad
  \infer{\cst{\app{M}{(\ltriv{R}{V}{m})}}{\ltriv{R}{(\app{M}{V})}{m}}}
    {}
\]
\[
  \infer{\cst{\ltriv{(\ltriv{R}{V}{m})}{N}{n}}{\ltriv{R}{(\ltriv{V}{N}{n})}{m}}}{}
\]
\[
  \infer{\cst{\lpair{x}{y}{(\ltriv{R}{V}{m})}{N}{n}}{\ltriv{R}{(\lpair{x}{y}{V}{N}{n})}{m}}}{}
\]
\[
  \infer{\cst{\ldown{m}{n}{x}{(\ltriv{R}{V}{m})}{N}}{\ltriv{R}{(\ldown{m}{n}{x}{V}{N})}{m}}}{}
\]
\[
  \infer{\cst{\ifc{(\ltriv{R}{V}{m})}{N_1}{N_2}}{\ltriv{R}{(\ifc{V}{N_1}{N_2})}{m}}}{}
\]
The cases for the pair elimination, down elimination, and if constructs are similar.

\section{Select Cases for Theorems}
\label{app:proofs}
\begin{theorem}[Type Preservation of Circuit-Normalization Semantics]
  If \ $\tp{\Gamma_{WCQ}}{m}{M}{A}$ and $\cst{M}{M'}$, then $\tp{\Gamma_{WCQ}}{m}{M'}{A}$.
\end{theorem}
\begin{proof}
  By induction on the derivation of $\cst{M}{M'}$.
 \begin{case}[General Congruence Case]
  \[
    \infer{\cst{\down{n}{m}{M}}{\down{n}{m}{M'}}}
    {\cst{M}{M'}}
  \]
  Then $\tp{\Gamma_{WCQ}}{m}{\down{n}{m}{M}}{{{\downarrow^n_m} A_n}}$ by assumption and $\Gamma_{WCQ} \geq n$ and $\tp{\Gamma_{WCQ}}{n}{M}{A_n}$ by inversion on ${\downarrow} I$. We obtain $\tp{\Gamma_{WCQ}}{n}{M'}{A_n}$ by the induction hypothesis, and thus $\tp{\Gamma_{WCQ}}{m}{\down{n}{m}{M'}}{{{\downarrow^n_m} A_n}}$ by ${\downarrow} I$.
\end{case}
\begin{case}[$\beta$ Case]
  \[
    \infer{\cst{\app{(\lam{x}{M})}{V}}{[V/x]M}}{}
  \]
  Then $\tp{\Gamma_{WCQ}}{m}{\app{(\lam{x}{M})}{V}}{B}$ for some $m$ by assumption, and thus $\tp{\Gamma_{WCQ}, x : A}{m}{M}{B}$ and $\tp{\Gamma_{WCQ}}{m}{V}{A}$ for some $A$ by inversion on $app$ and $lam$ on the lefthand premise. We obtain $\tp{\Gamma_{WCQ}}{m}{[V/x]M}{B}$ by the substitution lemma (\Cref{lem:subst}).
\end{case}
\begin{case}[$C$-Specific Congruence Case]
  \[
    \infer{\cst{\fixc{x}{M_C}}{\fixc{x}{M'_C}}}
    {\cst{M_C}{M'_C}}
  \]
  Then $\tp{\cdot}{C}{\fixc{x}{M_C}}{A}$ for some $A$ by assumption, and thus $\tp{x : {\dwc} {\ucw} A}{C}{M_C}{A}$ by inversion on $\mathit{fix}$.
  We obtain $\tp{x : {\dwc} {\ucw} A}{C}{M'_C}{A}$ by the induction hypothesis, and thus $\tp{\cdot}{C}{\fixc{x}{M'_C}}{A}$ by $\mathit{fix}$.
  \qedhere
\end{case}
\end{proof}
 
\begin{theorem}[Progress of Circuit-Generation Semantics]
If \ $\tp{\Gamma_{WCQ}}{m}{M}{A}$, then either $M$ is a value or there exists some $M'$ such that $\cst{M}{M'}$.
\end{theorem}
\begin{proof}
  By induction on the derivation of $\tp{\Gamma_{WCQ}}{m}{M}{A}$. We give a proof for the application case, which covers most of the arguments used in the other cases.
  \begin{case}
  \[
    \infer[app]{\tp{\GammaIII}{m}{\appm{M}{N}{m}}{B_m}}
    {\tp{\Gamma_1}{m}{M_m}{A_m \multimap_m B_m}
    &\tpm{\Gamma_2}{m}{N}{A}}
  \]
  Terms $M_m, N_m$ are either values or step to some $M_m', N_m'$, respectively, by induction hypothesis.
  We proceed by simultaneously casing on whether $M_m$ and $N_m$ are normal or not, and if they are, their syntaxes as given by the normal forms lemma (\Cref{lem:compile-norm}).
  We give several representative subcases:
  \begin{subcase}[$M_m = \lam{x}{N_m}$ and $m \neq C$ or $M_m = \lam{x}{V_C}$]
    $\cst{\appm{M}{N}{m}}{[N_m/x]M_m}$ by the $\beta$ rule for application.
  \end{subcase}
  \begin{subcase}[$M_m = g$ and $N_m = K_Q$ or $N_m = R_Q$]
    $\app{g}{K_Q}$ and $\app{g}{R_Q}$ are normal.
  \end{subcase}
  \begin{subcase}[$M_m = g$ and $N_m = \ltrivq{R_Q}{V_Q}$]
    $\cst{\appm{M}{N}{m}}{\ltrivq{R_Q}{(\appm{g}{V}{Q})}}$ by commuting conversion. Similar for $N_m = \lpairq{x}{y}{R_Q}{V_Q}$.
  \end{subcase}
  \begin{subcase}[$M_m$ steps to some $M_m'$]
    $\cst{\appm{M}{N}{m}}{\appm{M'}{N}{m}}$ by the congruence rule for application. Similar for when $M_m$ is normal and $N_m$ steps to some $N_m'$.
    \qedhere
  \end{subcase}
  \end{case}
\end{proof}

\end{document}
